\documentclass[aps,english,prx,floatfix,amsmath,superscriptaddress,tightenlines,twocolumn,nofootinbib]{revtex4-1}
\usepackage{mathtools, amssymb}
\usepackage{graphicx}
\usepackage{intcalc}
\usepackage[caption=false]{subfig}
\usepackage[shortlabels]{enumitem}

\usepackage{multirow}
\usepackage{hhline}
\usepackage{algorithmicx}
\usepackage{algorithm}
\usepackage[noend]{algpseudocode}

\usepackage{soul}
\usepackage[utf8]{inputenc}
\usepackage{xcolor}
\usepackage{amsthm}
\usepackage{float}
\usepackage{comment}
\usepackage{qcircuit}

\newcounter{protocol}
\makeatletter
\newenvironment{protocol}[1][htb]{
  \let\c@algorithm\c@protocol
  \renewcommand{\ALG@name}{Protocol}
  
  \begin{algorithm}[#1]
  }{\end{algorithm}
}
\makeatother

\newcommand{\ket}[1]{|#1\rangle}

\definecolor{maroon}{rgb}{144,12,63}
\definecolor{darkblue}{rgb}{27,12,144}
\definecolor{mypurple2}{RGB}{170,0,255}
\definecolor{myred}{RGB}{255, 0, 90}
\definecolor{mycyan}{RGB}{0, 191, 255}

\theoremstyle{definition}

\newtheorem{claim}{Claim}

\usepackage{hyperref}
\hypersetup{colorlinks=true,linkcolor=myred,citecolor=black,urlcolor=black}
\usepackage[all]{hypcap}

\newcommand{\Q}{{\mathcal{Q}}}
\newcommand{\Gc}{{G_{\mathrm{c}}}}
\newcommand{\Gbcc}{{G_{\mathrm{bcc}}}}

\begin{document}
\title{Fault-tolerant qubit from a constant number of components}
\author{Kianna Wan}
\affiliation{Stanford Institute for Theoretical Physics, Stanford University, Stanford, CA 94305, USA}
\author{Soonwon Choi}
\affiliation{Department of Physics, University of California Berkeley, Berkeley, CA 94720, USA}
\author{Isaac H. Kim}
\affiliation{Department of Computer Science, UC Davis, Davis, CA 95616, USA}
\affiliation{School of Physics, The University of Sydney, Sydney, NSW 2006, Australia}
\author{Noah Shutty}
\affiliation{Stanford Institute for Theoretical Physics, Stanford University, Stanford, CA 94305, USA}
\author{Patrick Hayden}
\affiliation{Stanford Institute for Theoretical Physics, Stanford University, Stanford, CA 94305, USA}
\date{\today}

\begin{abstract}
With gate error rates in multiple technologies now below the threshold required for fault-tolerant quantum computation, the major remaining obstacle to useful quantum computation is scaling, a challenge greatly amplified by the huge overhead imposed by quantum error correction itself. We propose a fault-tolerant quantum computing scheme that can nonetheless be assembled from a small number of experimental components, potentially dramatically reducing the engineering challenges associated with building a large-scale fault-tolerant quantum computer. Our scheme has a threshold of $0.39\%$ for depolarising noise, assuming that memory errors are negligible. In the presence of memory errors, the logical error rate decays exponentially with $\sqrt{T/\tau}$, where $T$ is the memory coherence time and $\tau$ is the timescale for elementary gates. Our approach is based on a novel procedure for fault-tolerantly preparing three-dimensional cluster states using a single actively controlled qubit and a pair of delay lines. Although a circuit-level error may propagate to a high-weight error, the effect of this error on the prepared state is always equivalent to that of a constant-weight error. We describe how the requisite gates can be implemented using existing technologies in quantum photonic and phononic systems. With continued improvements in only a few components, we expect these systems to be promising candidates for demonstrating fault-tolerant quantum computation with a comparatively modest experimental effort. 
\end{abstract}
\maketitle

\section{Introduction}
In recent years, significant experimental progress has been made towards building a large-scale quantum computer. In platforms such as superconducting qubits and trapped ions, the error rates for small systems have been successfully suppressed below the threshold error rate of the surface code~\cite{Barends2014,Harty2014,Ballance2016}. Using newly developed techniques for neutral atoms trapped in optical tweezer arrays, the coherence time, gate fidelity, and read-out fidelity for large assemblies of qubits are being rapidly improved~\cite{Kaufman2018control,Thompson2019narrow,Endres2019repeated,levine2019parallel}. These advances give us hope that we will one day be able to perform fault-tolerant quantum computation by scaling up these systems while maintaining low error rates.

However, the scalability of leading approaches remains an important issue. Current estimates suggest that the engineering effort needed to build even a \emph{single} logical qubit with logical error rate low enough for useful quantum computation could be enormous~\cite{Fowler2009}. Quantum algorithms with practical ramifications can involve applying at least $\sim\!10^{8}$ logical gates to $\sim\!100$ logical qubits~\cite{Bauer2016,Babbush2018}. To ensure that the outcome of the computation is correct with high probability, the logical error rate would then need to be below $\sim\!10^{-8}$. Based on the sub-threshold error scaling in Ref.~\cite{Fowler2012}, this would require at least $\sim\!400$ physical qubits per logical qubit if the physical error rate is half the threshold.

Manufacturing, calibrating and controlling physical qubits in such large numbers will be tremendously difficult. The fabrication process for components of solid-state quantum devices, such as quantum dots or superconducting circuits~\cite{Barends2014}, is inevitably imperfect, leading to variations in the properties of individual qubits and their interactions. Even in systems where qubits are encoded in identical particles, \textit{e.g.}, trapped ions~\cite{Harty2014,Ballance2016,monroe2013scaling} or neutral atoms~\cite{Kaufman2018control,Thompson2019narrow,Endres2019repeated,levine2019parallel}, experimental control parameters such as the strengths of laser excitation pulses or trapping potentials may exhibit inhomogeneity. Thus, in order to control these qubits with high fidelity, an experimental system needs to be accurately calibrated across the entire quantum computer. In superconducting circuits, for instance, inhomogeneity is unavoidable, and stray couplings between ideally independent qubits are an experimental fact of life that must be mitigated through control logic (see \textit{e.g.}, \cite{Barends2014}.) The difficulty of doing so increases significantly with the number of qubits~\cite{arute2019quantum}.

To circumvent these challenges, we propose a novel approach to fault-tolerant quantum computation, in which a well-protected logical qubit can be built using only a handful of experimental components. Consequently, the engineering effort required to develop the computer's components can be significantly reduced, potentially opening a simpler and more easily scalable route to fault-tolerant quantum computation. At a high level, our approach succeeds by shedding the limitations implicit in two assumptions that usually guide fault-tolerant circuit design: first, that the computer's qubits are all of the same type so are fairly homogeneous, and second, that good fault-tolerant gates should not propagate errors.

Specifically, we construct a fault-tolerant protocol for generating the three-dimensional cluster state of Ref.~\cite{Raussendorf2006}, using which universal fault-tolerant computation can be performed via adaptive single-qubit measurements. While there are already well-known procedures for preparing this state~\cite{Raussendorf2006,Cluster1,Cluster2,Cluster3,Cluster4,Cluster5}, our method has the advantage of being compatible with a much simpler experimental setup than what was originally envisaged in Refs.~\cite{Raussendorf2006,Raussendorf2007a,Raussendorf2007}. We take an approach similar to existing proposals for building large one- and two-dimensional cluster states using a small number of physical components~\cite{Economou2010_cluster2D,Yokoyama2013,Junichi2016,Pichler2017,Asavanant2019}. However, while two-dimensional cluster states are universal for quantum computation, they are not known to support \textit{fault-tolerant} quantum computation.\footnote{More precisely, two-dimensional cluster states of \textit{unprotected physical qubits} are not known to be a particularly useful resource for fault-tolerant quantum computation. Two-dimensional cluster states can be used to perform local gates on a one-dimensional array of qubits~\cite{Raussendorf2003}, for which fault-tolerant quantum computing schemes have been developed~\cite{Gottesman1999,Stephens2009}.  However, the threshold (in the circuit model) is likely prohibitively low (estimated to be $10^{-5}$~\cite{Stephens2009}).} The step from universality to fault-tolerance is not obvious and, in fact, quite surprising considering the architecture of the system.

Our protocol is built around a special ancilla qubit, $\Q$, which interacts sequentially with a stream of data qubits propagating through a delay line. These data qubits are encoded in degrees of freedom sharing a common physical implementation, \textit{e.g.}, different temporal modes of photons or phonons in a waveguide. The only interactions are between $\Q$ and data qubits (and not between data qubits themselves), and these interactions are fixed and periodic, requiring a modest amount of calibration. We show, moreover, that all of the operations required in our protocol can be implemented using existing technologies in quantum photonic and phononic systems.

To demonstrate fault-tolerance, we analyse the robustness of our protocol against both circuit errors and memory errors. We use a standard depolarising model to describe circuit errors, which are associated with imperfect gates, measurements, and state initialisation. Memory errors refer to errors that occur while qubits are idle, for which we study the effect of dephasing and qubit loss.

In the absence of memory errors, there is a threshold of $0.39\%$ for the circuit error rate, below which the logical error can be arbitrarily suppressed by increasing the number of physical qubits. In the presence of memory errors, the logical error rate cannot be arbitrarily suppressed. However, provided that the circuit error rate is below threshold, the logical error rate decays rapidly with the inverse of the memory error rate. More precisely, suppose that the coherence time of the data qubits is lower-bounded by $T$. Then, for a sufficiently large but finite $T$, the logical error rate can be made exponentially small in $\sqrt{T/\tau}$. Here, $\tau$ is the inverse of the frequency with which gates are applied, which is ultimately limited by the timescale for interactions between $\Q$ and data qubits. The number of logical gates that can be reliably executed will therefore scale exponentially with $\sqrt{T/\tau}$.

A large separation between $T$ and $\tau$ is often observed in certain experimental platforms, such as trapped ions or neutral atoms utilising atomic clock transitions~\cite{monroe2013scaling,Kaufman2018control,Thompson2019narrow,Endres2019repeated}. Indeed, because of the strict separation in the roles of $\Q$ and the data qubits, maximising the ratio $T/\tau$ while maintaining high gate fidelity is an invitation to design a hybrid system consisting of two types of qubits with different physical substrates. That is the context in which we expect our scheme to be the most promising. Photonic~\cite{Tamura2018} and phononic~\cite{PhysRevLett.121.040501} delay lines are known to be good quantum memories, and can be coupled to controllable qubits capable of playing the role of $\Q$. 

To illustrate the potential of our scheme, suppose that memory errors are dominated by loss. Then, if the circuit error rate is $10^{-3}$---an aspirational but realistic target---our protocol can in principle attain a logical error rate of $10^{-8}$ for $\tau/T \approx 1.4 \times 10^{-5}$, and $10^{-15}$ for $\tau/T \approx 3.2\times 10^{-6}$. Although these numbers are beyond the reach of current experiments, these estimates suggest that extremely low logical error rates can be achieved by improving a very small number of experimental components. In particular, if the operations involving $\Q$ can be calibrated such that circuit error rate is below the threshold value of $0.39\%$, incremental improvements of a \emph{single} component---the delay line---can lead to drastic reductions in the logical error rate. 

Although our scheme was primarily motivated by the aforementioned experimental considerations, it also has a novel feature that is counterintuitive from the point of view of fault-tolerance. The design of fault-tolerant protocols usually aims to prevent the propagation of single-qubit errors to many qubits. This is achieved, naturally enough, by applying gates that do not spread errors, \textit{e.g.}, transversal gates, or ``long" gates that are interspersed with error correction steps, such as in lattice surgery~\cite{Horsman2011}. In all of these methods, one actively avoids interacting one qubit with many others in a code block, since errors occurring on that qubit could propagate to the others, exceeding the error-correcting capabilities of the code.

In our protocol, we are actually deliberately taking this seemingly ill-advised approach: a single qubit ($\Q$) is coupled to \emph{every} data qubit. The depth of the circuit scales linearly with the number of data qubits, and no error detection or correction is performed during the process. Nevertheless, the procedure is fault-tolerant in that any single-qubit error occurring in the circuit results in a constant-weight error on the final state. An interesting subtlety is that even though a single-qubit circuit-level error can in general be propagated by the subsequent gates to a highly nonlocal error, this nonlocal error is always equivalent under stabilisers of the prepared cluster state to some geometrically local error. More generally, we show that any $m$-qubit circuit-level error results in at most $m$ geometrically local errors on the final state. 

To summarise, our proposal and analysis indicate that fault-tolerant quantum computation could be achieved through the incremental improvement of a small number of key components, avoiding most of the systems engineering challenges inherent in leading approaches. This is possible because of three important features of our scheme. First, it only requires manufacturing and calibrating a constant number of experimental components, independent of the number of data qubits. Second, there are readily available experimental platforms that can realise our protocol. Third, any constant-weight error occurring during our protocol results in a constant-weight error on the prepared cluster state.

The rest of this paper is structured as follows. We provide the necessary background and a summary of our main results in Section~\ref{section:background_summary}. We then present a hardware-independent description of our protocols in Section~\ref{sec:cluster}. In Section~\ref{sec:error}, we analyse how errors propagate through our circuits, and numerically calculate thresholds for the circuit error rate. In Section \ref{sec:experiment}, 
we outline possible experimental implementations of our proposal in photonic and phononic systems. In Section~\ref{sec:delay_line}, we study the effect of memory errors, estimating the logical error rates we can expect to achieve in various experimental platforms. We conclude with a discussion in Section~\ref{sec:discussion}.

\section{Summary}\label{section:background_summary}
We start by briefly reviewing the subject of fault-tolerant measurement-based quantum computation using cluster states in Section~\ref{section:background}, focusing on the aspects that are relevant to this paper. We then summarise our main results in Section~\ref{section:summary}.

\subsection{Background}\label{section:background}
The \textit{cluster state} $\ket{\psi_G}$ corresponding to an undirected graph $G = (V,E)$ is defined as\footnote{States of the form of Eq.~\eqref{cluster_def} are also referred to as \textit{graph states} in the literature.}
\begin{equation} \label{cluster_def}
\ket{\psi_G} \coloneqq \left[\prod_{(i,j) \in E} Z_{i,j}\right]\bigotimes_{i' \in V}\ket{+}_{i'},
\end{equation}
where each vertex $i \in V$ is identified with a qubit, and $Z_{a,b}$ denotes the controlled-$Z$ gate on qubits $a$ and $b$. The stabilisers of $\ket{\psi_G}$ are generated by $\{S_i: i \in V\}$, where~\cite{Raussendorf2003}
\begin{equation} \label{stabilisers} S_i \coloneqq X_i \prod_{j: (i,j) \in E}Z_j. \end{equation} Here, $X_a$ and $Z_a$ denote Pauli $X$ and $Z$ on qubit $a$. 

The importance of cluster states in the theory of fault-tolerant quantum computation was established by the seminal works of Raussendorf, Harrington, and Goyal~\cite{Raussendorf2006,Raussendorf2007a, Raussendorf2007}, which demonstrated that universal fault-tolerant quantum computation can be performed via single-qubit measurements on a particular cluster state. This cluster state corresponds to the body-centered cubic (bcc) lattice shown in Figs.~\ref{fig:bcc} and~\ref{fig:bcc_labels}. Their scheme (for constructing this cluster state and extracting the syndrome) boasts a high threshold of $p_{\textrm{th}}\approx 0.58\%$~\cite{Barrett2010}
under the standard depolarising model for circuit errors, making it one of the most promising approaches for building a large-scale quantum computer.

To prepare the cluster state $\ket{\psi_{\Gbcc}}$ corresponding to the bcc lattice $\Gbcc$, Refs.~\cite{Raussendorf2006,Barrett2010} consider a simple constant-depth circuit, which follows directly from Eq.~\eqref{cluster_def}. Each qubit is initialised in the state $\ket{+}$, and the controlled-$Z$ gates in Eq.~\eqref{cluster_def} for $G = \Gbcc$ are applied in four layers. It is straightforward to see that any single-qubit error in this circuit propagates to a constant-weight error on $\ket{\psi_{\Gbcc}}$.
Together with the fact that $\ket{\psi_\Gbcc}$ is a foliation of the surface code~\cite{Bolt2016}, this implies that there is a finite threshold for the circuit error rate below which the logical error rate decays exponentially with the system size.

Given a cluster state on an $L \times L \times N$ bcc lattice, one can perform fault-tolerant quantum computation by adaptively measuring the qubits in one of three bases (the eigenbases of the operators $X, Z,$ and $e^{i\frac{\pi}{8}Z} X e^{-i\frac{\pi}{8}Z}$),
depending on the logical gates that are to be executed. Note that a qubit can be measured before the full cluster state has been prepared, provided that all of the controlled-$Z$ gates in Eq.~\eqref{cluster_def} involving that qubit have been applied. Thus, the cluster state could alternatively be prepared and measured in such a way that only $O(L^2)$ physical qubits are in use at any given time. Roughly speaking, $L$ determines the number of logical qubits that can be encoded and the distance of the underlying code, while $N$ is related to the length of the logical computation. We refer the reader to Refs.~\cite{Fowler2009,Raussendorf2006,Raussendorf2007,Raussendorf2007a} for further details.

Even though fault-tolerant computation can be in principle  performed on such a cluster state, in this paper, we focus on realising a fault-tolerant quantum memory. In particular, we consider using a cluster state on an $L \times L \times L$ bcc lattice to store a single logical qubit. From the perspective of quantum error correction, this cluster state can be viewed as a space-time history of the surface code~\cite{Kitaev2003,Bravyi1998} with $L$ rounds of syndrome measurements, the bottom and the top boundaries of the cluster state corresponding to the surface codes at the initial and the final step of the error-correction protocol. Our estimates for the logical error rate, which decays exponentially with $L$ under local noise models (cf. Sections~\ref{sec:error} and~\ref{sec:delay_line}), quantifies the probability that there is a logical bit or phase flip between the bottom and the top layer. 

The leading architecture for implementing this scheme is based on a two-dimensional array of physical qubits~\cite{Raussendorf2007,Raussendorf2007a,Fowler2012}. This approach suffers from an important practical problem, however. The \emph{space overhead}, which is the ratio between the number of physical qubits and the number of logical qubits, is quite large in practice. For instance, the space overhead for running Shor's algorithm~\cite{Shor1997}, assuming a physical error rate of $10^{-3}$, is estimated to be at least a few hundred~\cite{Fowler2009,Gidney2019}. Thus, building even a single logical qubit with low enough error rate will require hundreds if not thousands of physical components. Moreover, these components will need to be carefully calibrated to ensure that the physical error rates across all of the qubits are sufficiently low. While this is not impossible, it certainly requires a Herculean effort.

\subsection{Main results}\label{section:summary}

Generally speaking, large space overhead is undesirable because the effort to build a fault-tolerant quantum computer may grow proportionately with the number of physical qubits. However, for the purpose of assessing the feasibility of a given architecture, it is important to distinguish the mathematical definition of space overhead from the engineering difficulty of building a quantum computer. We believe that a useful figure of merit for the latter is the \emph{component overhead}, which is the number of basic experimental components used to build a single logical qubit. Of course, the precise definition of ``experimental component" depends on the degrees of freedom that encode the quantum information. Once those degrees of freedom are identified, one can compare different protocols in terms of the required experimental components. This information can be related more directly to the feasibility of the protocol. 

Component overhead can be an informative metric because the basic building blocks that constitute a large-scale fault-tolerant quantum computer may be difficult to mass manufacture. Even though there are several experiments that report error rates below the thresholds of various fault-tolerant quantum computing schemes~\cite{Raussendorf2006,Raussendorf2007,Raussendorf2007a,Fowler2009,Knill2004},
these numbers are often obtained in a manner that is incompatible with scalability. This is due to the practical reality that when the components are manufactured, they have sample-to-sample variations which lead to imperfect gates. Often, the reported numbers come from the very best of those samples, but if the variation is not negligible, many of the other samples will generally suffer from higher error rates. Therefore, given that high-quality components are difficult to come by, scalable fault-tolerant quantum computing protocols should aim to minimise the number of such components. 

Motivated by this observation, we construct simple abstract protocols for fault-tolerant quantum computation that are amenable to extremely low component overhead. We also present concrete experimental proposals for realising the protocols using a single transmon qubit interacting with a stream of phonons, or alternatively, an atom interacting with a stream of photons. Our protocols may be applicable more generally, \textit{e.g.,} to systems consisting of ions or neutral atoms. There are two distinguishing features of all these systems that are crucial. First, the degrees of freedom that encode the quantum information are either identical by nature or can be made to be nearly identical. Second, the qubits have long coherence times, leading to low memory error rates.

For systems that fulfill these conditions, we describe a simple method for preparing cluster states corresponding to the bcc lattice. A schematic illustration of the setup is given in Fig.~\ref{fig:protocol_schematic}. The procedure involves two types of qubits, a single actively controlled qubit $\Q$ and a large number of data qubits. Each data qubit interacts with $\Q$ several times, and these interactions are separated by time delays determined by the size of the bcc lattice. The data qubits do not interact with each other. The gates applied between $\Q$ and the data qubits are specified in Section~\ref{sec:cluster}, and experimental techniques for realising these gates are described in Section~\ref{sec:experiment}. The procedure is an extension of the photonic machine gun proposal of Ref.~\cite{Lindner2009} and variants thereof~\cite{Pichler2017}. These works advocated methods for creating cluster states on one- and two-dimensional lattices, respectively, neither of which are known to be useful resources for fault-tolerant quantum computation. In contrast, our protocol prepares the cluster state on the bcc lattice, which (as discussed in Section~\ref{section:background}) can be straightforwardly used as a resource state for fault-tolerant measurement-based quantum computation.

\begin{figure}
    \centering
    \includegraphics[width=0.9\columnwidth]{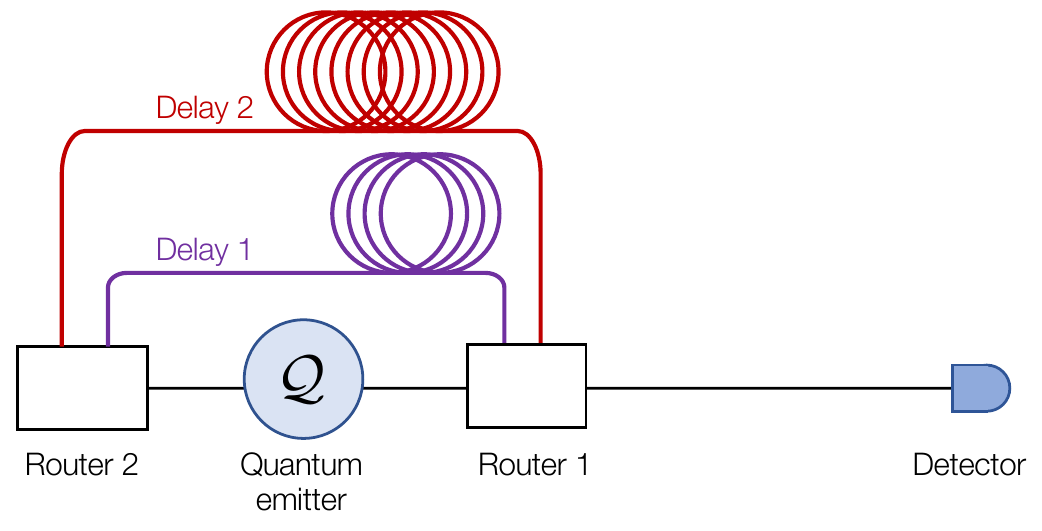}
    \caption{A schematic illustration of the apparatus for implementing our protocols. Propagating modes, \textit{e.g.}, photons or phonons, are stored in delay lines between interactions with the control qubit $\Q$ (which, in the context of the experimental platforms considered in Section~\ref{sec:experiment}, is a quantum emitter). Each qubit is measured after a constant number of interactions with $\Q$. See Section~\ref{sec:cluster} for the abstract description of our protocols, and Section~\ref{sec:experiment} for experimental details. 
    \label{fig:protocol_schematic}}
\end{figure}

Independent of the precise sequence of gates between the control qubit $\Q$ and the data qubits, any protocol of the form depicted in Fig.~\ref{fig:protocol_schematic} is at risk of being strongly susceptible to noise. There are two potential sources of concern. The first is that $\Q$ interacts with every single data qubit, without any intermediate syndrome measurements being performed. This creates the danger that an error occurring on $\Q$ could propagate to all of the data qubits that subsequently interact with $\Q$. The second issue is that there is a time delay between successive interactions of the same data qubit with $\Q$. For generating an $L \times L \times N$ bcc lattice, the total time delay is proportional to $L^2$. 
Thus, the total memory error accumulated during these time delays may be significant.

The first of these is actually a non-issue. As discussed in Section~\ref{sec:error}, an important feature of our protocols is that even though single-qubit errors, including those on $\Q$, may propagate to highly nonlocal errors, the effect of these errors on the prepared cluster state is always equivalent to that of geometrically local errors. Hence, using the standard depolarising noise model for circuit errors and the usual minimum-weight perfect matching decoder, there is a finite threshold for the circuit error rate. We find threshold values of $0.23\%$ and $0.39\%$, depending on the details of the protocol (see Sections~\ref{sec:prep_3d} and~\ref{sec:threshold_result}). Therefore, if memory errors are negligible, the logical error rate can be arbitrarily suppressed by increasing $L$.

In contrast, for non-negligible memory error rates, the logical error rate cannot be made arbitrarily small, since increasing $L$ also leads to an increase in the total error incurred during the time delays. We study these effects in Section~\ref{sec:delay_line} by assuming a nonzero error rate $\eta$ per time step. As long as the circuit error rate is below threshold, we argue that by judiciously choosing $L$, the logical error rate can be made exponentially small in $\eta^{-{1}/{2}}$. We perform extensive numerical simulations, whose results show excellent agreement with this prediction. Since the logical error rate decays significantly faster than $\eta$ for small values of $\eta$, the effect of memory error can be mitigated.

The fact that our scheme leads to small but not arbitrarily small logical error rates is reminiscent of the fault-tolerant quantum computing schemes using anyons~\cite{Kitaev2003,Nayak2008,Alicea2011,Nakamura2020} or $0$-$\pi$ qubits~\cite{Kitaev2006,Manucharyan2009,Brooks2013}. In these approaches, the logical error rate is exponentially small in some large physical parameter. In ours, this parameter is $\eta^{-{1}/{2}}$.

\section{Cluster state preparation} \label{sec:cluster}

In this section, we present a general algorithm for preparing cluster states associated with arbitrary graphs $G = (V,E)$ (Section~\ref{sec:general prep}). We then apply this algorithm in two different ways to prepare cluster states on the bcc lattice of Ref.~\cite{Raussendorf2006} (Section~\ref{sec:prep_3d}). 

The standard procedure for preparing cluster states is to initialise each qubit in the $\ket{+}$ state and apply controlled-$Z$ gates according to Eq.~\eqref{cluster_def}. Since a controlled-$Z$ gate between qubits $a$ and $b$ is required for every $(a,b) \in E$, this approach involves $|E|$ distinct gates. All of these gates must be carefully calibrated and implemented, making the experimental realisation of this protocol daunting. 

In contrast, our protocols
bypass the need to calibrate and implement a large number of physically distinct operations, allowing for simple experimental realisations, as explained in Section~\ref{sec:experiment}. In our algorithm, there is a single ancilla $\Q$ that interacts with the data qubits (which correspond to the vertices $V$ of $G$) one by one. Physically, $\Q$ is an actively controlled qubit, while the data qubits are identical degrees of freedom (\textit{e.g.}, phonons or photons generated from a single source, ions, or neutral atoms) that passively interact with $\Q$. The data qubits do not ever need to interact with each other. In this setting, one can simply tune a constant number of interactions between the controllable qubit and the physical system representing the data qubits to calibrate \emph{every} gate. 

\subsection{Algorithm for arbitrary graphs} \label{sec:general prep}

In this subsection, we provide an algorithm, Algorithm~\ref{alg1}, for preparing cluster states $\ket{\psi_G}$ on arbitrary graphs $G$. The correctness proof for this algorithm is given in Appendix~\ref{appendix:correctness_measurement_based}.

First, we define some notation. Here and throughout the paper, $H_a$ denotes the Hadamard gate acting on qubit $a$, and $P_a$ the Pauli operator $P \in \{X,Y,Z\}$ on qubit $a$. We write $X_{a,b}$ to represent the controlled-$X$ gate with control qubit $a$ and target qubit $b$, and $Z_{a,b}$ the controlled-$Z$ gate between qubits $a$ and $b$, with \[ \text{$Z_{\Q,i} \equiv I$ \enspace for any $i \not\in V$}. \] 
We also use the convention that the operators $A_j$ in the product $\prod_{j=1}^k A_j$ are ordered as $A_kA_{k-1}\dots A_1$, and that an empty product of operators acts as the identity. 

The main idea behind Algorithm~\ref{alg1} is to generate progressively larger cluster states related to subgraphs of $G = (V,E)$ by adding in one qubit at a time. Specifically, let $n \coloneqq |V|$ be the number of data qubits and fix an ordering of the qubits by labelling them from $1$ to $n$. For a given ordering, the qubit labelled $1$ is added first, followed by the qubit labelled $2$, and so on.

\begin{algorithm}[H] \caption{prepare the cluster state $\ket{\psi_G}$ given a graph $G = (V, E)$ with $V = [n]$} \label{alg1}
\begin{algorithmic}[1]
    \State initialise $\Q$ in the state $\ket{+}$ \label{alg1: initialise Q}
    \For{$j=1$ to $n$} \State initialise qubit $j$ in the state $\ket{0}$ \label{alg1: initialise j}
    \State apply $H_\Q X_{\Q,j}\prod\limits_{\substack{i< j-1 \\(i,j) \in E}}Z_{\Q,i}$ \label{alg1: Bj}
    
    \textit{\hspace{1em} // the $Z_{\Q,i}$ gates may be applied in any order}
    \If{$(j,j+1) \not\in E$ or $j=n$} \label{alg1: if1} 
    \State measure $\Q$ in the $Z$-basis
    \label{alg1: measure}
    \If{the outcome is $\ket{1}$} \State apply $Z_j$ \label{alg1: Zj}
    \EndIf
    \State re-initialise $\Q$ in $\ket{+}$ \hfill \textit{\hspace{1em}// not necessary for $j = n$} \label{alg1: reset}
    \EndIf
    \EndFor
\end{algorithmic}
\end{algorithm}

To explain the algorithm, it will be convenient to define graphs $G[k]'$ as follows. For each $k \in [n] \coloneqq \{1,\dots, n\}$, let $E[k]$ denote the set of edges in the subgraph of $G$ induced by the vertex subset $[k]$, \textit{i.e.},
\begin{equation} \label{E[k]} E[k] \coloneqq \{(i,j) \in E: i,j \in [k]\}. \end{equation}
Then, let
\begin{equation} \label{G[k]'}
    G[k]' \coloneqq \left([k] \cup \{\Q\}, E[k] \cup \{(Q,k)\}\right)
\end{equation}
be the graph with vertex set $[k] \cup \{\Q\}$ and edge set $E[k] \cup \{(\Q,k)\}$. See Fig.~\ref{fig:G[k]'}(a) for an example. Note that
$G[k]'$ differs from the subgraph of $G$ induced by $[k]$ only in that it contains an extra vertex $\Q$ and an extra edge $(\Q, k)$. 
\begin{figure*}
    \centering
    \includegraphics[width=1.7\columnwidth]{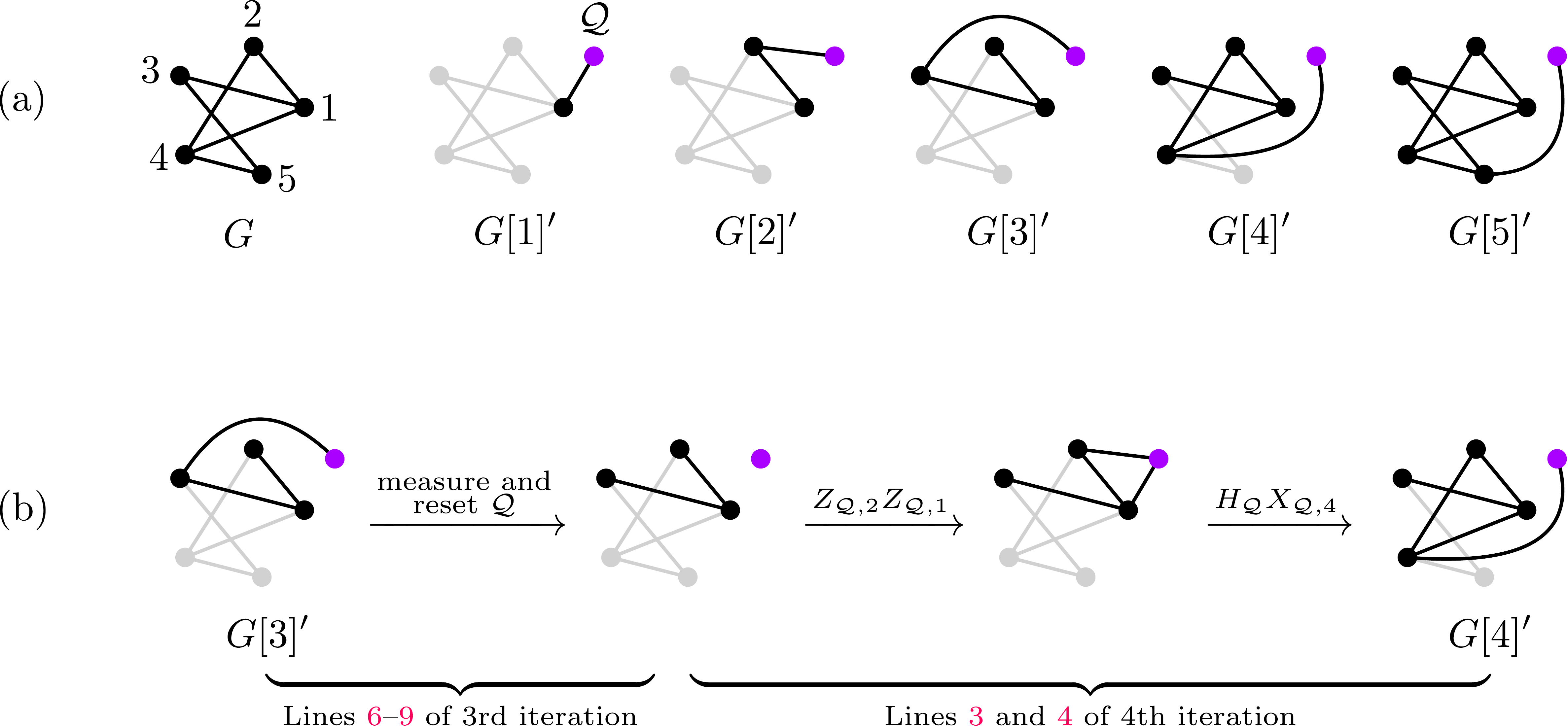}
    \caption{(a) An example of a graph $G = (V,E)$ and an ordering of its $n=5$ vertices, along with the corresponding graphs $G[k]'$ [cf.~Eq.~\eqref{G[k]'}] for $k \in [5]$. (The grey vertices and edges in $G[k]'$ for $k \in [4]$ are not part of the graph.) For each $k \in [n]$, the state of $\Q$ and the first $k$ data qubits at the end of Line~\ref{alg1: Bj} in the $k$th \textbf{for} loop iteration in Algorithm~\ref{alg1} is the cluster state $\ket{\psi_{G[k]'}}$. (b) Since $(3, 4) \not\in E$ in this example, the \textbf{if} condition of Line~\ref{alg1: if1} is satisfied in the $j=3$ iteration, and Lines~\ref{alg1: measure}--\ref{alg1: reset} are executed. Then, in the next iteration ($j = 4$), Lines~\ref{alg1: initialise j} and~\ref{alg1: Bj} change the state to $\ket{\psi_{G[4]'}}$. Note that $H_\Q X_{\Q,i}$ has the effect of swapping the state of $\Q$ onto qubit $i$ and adding an edge between $\Q$ and $i$; see Eq.~\eqref{HQ identity} in Appendix~\ref{appendix:correctness_measurement_free}.}
    \label{fig:G[k]'}
\end{figure*}

The cluster state $\ket{\psi_{G[k]'}}$ corresponding to $G[k]'$ is defined via Eq.~\eqref{cluster_def}. In Appendix~\ref{appendix:correctness_measurement_based}, we prove that for each $k \in [n]$, after Line~\ref{alg1: Bj} in the $k$th iteration of the \textbf{for} loop has been executed, the state of $\Q$ and the first $k$ data qubits is $\ket{\psi_{G[k]'}}$. Thus, Algorithm~\ref{alg1} prepares $\ket{\psi_G}$ by introducing a new data qubit in each iteration,
sequentially generating $\ket{\psi_{G[1]'}}, \ket{\psi_{G[2]'}}, \dots, \ket{\psi_{G[n]'}}$. The main steps are illustrated schematically in Fig.~\ref{fig:G[k]'}(b).

Once we have the state $\ket{\psi_{G[n]'}}$ (at the end of Line~\ref{alg1: Bj} in the last iteration), the desired state $\ket{\psi_G}$ can be easily obtained. Since the only difference between the two states is that $\ket{\psi_{G[n]'}}$ has an extra edge between $\Q$ and $n$, \textit{i.e.}, $\ket{\psi_{G[n]'}} = Z_{\Q,n}\ket{+}_\Q\ket{\psi_G}$, we can either apply $Z_{\Q,n}$ or measure $\Q$ in the $Z$-basis (and apply $Z_n$ if the outcome is $\ket{1}$).

As shown in Appendix~\ref{appendix:correctness_measurement_based}, the purpose of applying $Z_j$ in Line~\ref{alg1: Zj} is to ``fix" the cluster state in the case where $\Q$ is measured in Line~\ref{alg1: measure} of the $j$th iteration and the outcome is $\ket{1}$. Observe that all of the necessary $Z_j$ corrections could be deferred to the end of the procedure, instead of being implemented immediately. Alternatively, the $Z_j$ need not be applied at all if we keep track of all of the measurement outcomes and the modified cluster state stabilisers in the subsequent computation.

Note that different orderings of the qubits (\textit{i.e.}, different assignments of the labels $1$ through $n$ to the vertices in $V$) give rise to different circuits via Algorithm~\ref{alg1}, but every such circuit correctly produces the same state $\ket{\psi_G}$. One may choose an ordering that is more conducive to experimental realisation of the algorithm. 
Furthermore, in the case where $G$ contains a Hamiltonian path, Algorithm~\ref{alg1} does not require any intermediate measurements of $\Q$. By ordering the qubits such that $(i,i+1) \in E$ for all $i \in [n-1]$, Lines~\ref{alg1: measure}--\ref{alg1: reset} are skipped in every iteration of the main loop, which simplifies the procedure.

\subsection{3D cluster states}
\label{sec:prep_3d}

In this subsection, we describe two protocols, Protocols~\ref{protA} and~\ref{protB}, for preparing cluster states on the bcc lattice $\Gbcc$ of Ref.~\cite{Raussendorf2006} [cf.~Figs.~\ref{fig:bcc} and~\ref{fig:bcc_labels}]. Protocol~\ref{protA} involves first using Algorithm~\ref{alg1} to prepare the cluster state on a cubic lattice, then measuring out certain qubits to obtain $\ket{\psi_\Gbcc}$. Protocol~\ref{protB} applies Algorithm~\ref{alg1} to $\Gbcc$ directly. We propose experimental implementations of both protocols in Section~\ref{sec:experiment}.

\begin{figure}[h]
    \centering
    \includegraphics{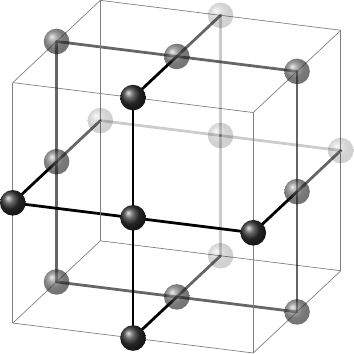}
    \caption{Elementary cell of the bcc lattice $G_{\mathrm{bcc}} = (V_{\mathrm{bcc}}, E_{\mathrm{bcc}})$.}
    \label{fig:bcc}
\end{figure}

These protocols have different strengths and weaknesses. Unlike Protocol~\ref{protB}, Protocol~\ref{protA} requires no intermediate measurements of the controllable qubit $\Q$, and is therefore expected to be simpler to implement. However, as we show in Section~\ref{sec:threshold_result}, the error threshold of Protocol~\ref{protA} is lower than that of Protocol~\ref{protB}.

\subsubsection{Protocol A} \label{sec:protA}
Protocol~\ref{protA} consists of two main steps. First, we use Algorithm~\ref{alg1} to prepare the cluster state $\ket{\psi_{\Gc}}$ on a certain cubic lattice $\Gc$, defined below, that contains $\Gbcc$ as a subgraph (Line~\ref{protA: 1}). Second, we obtain $\ket{\psi_\Gbcc}$ from $\ket{\psi_\Gc}$ by removing the qubits that are not in $\Gbcc$ via single-qubit $Z$-measurements (Lines~\ref{protA: for}--\ref{protA: correct}). 

The cubic lattice we consider is the graph $G_{\mathrm{c}} = (V_{\mathrm{c}},E_{\mathrm{c}})$ with vertex set $V_{\mathrm{c}} = [n]$ and edge set $E_{\mathrm{c}}$, defined for $L, M \in \mathbb{N}$ by
\begin{equation} \label{cubic edges}
\begin{aligned}
    E_{\mathrm{c}} &\coloneqq \{(i,i+1):i \in [n-1]\} \\
    &\quad \cup \{(i, i + L) : i \in [n-L] \} \\
    &\quad \cup \{(i,i+LM): i \in [n-LM]\}.
\end{aligned}
\end{equation}
If $n = LMN$ for some $N \in \mathbb{N}$, then $G_{\mathrm{c}}$ is a $L \times M\times N$ cubic lattice with shifted periodic boundary conditions; $\Gc$ differs from a standard cubic lattice with open boundary conditions only in that $\Gc$ has various additional edges between vertices on the boundary.

Note from Eq.~\eqref{cubic edges} that for every $i \in [n-1]$, $(i, i+1)$ is an edge in $\Gc$. Consequently, when we apply Algorithm~\ref{alg1} to $\Gc$, the \textbf{if} condition of Line~\ref{alg1: if1} is never satisfied and Lines~\ref{alg1: measure}--\ref{alg1: reset} are not executed, except in the very last iteration ($j = n$) of the \textbf{for} loop. Thus, Algorithm~\ref{alg1} reduces to a unitary circuit that prepares $\ket{\psi_{\Gc[n]'}}$, together with a single measurement of $\Q$ at the end to change $\ket{\psi_{\Gc[n]'}}$ to $\ket{\psi_{\Gc}}$. This circuit is shown in Fig.~\ref{fig:circ_example}.

\begin{figure*}
\includegraphics[width=0.9\textwidth]{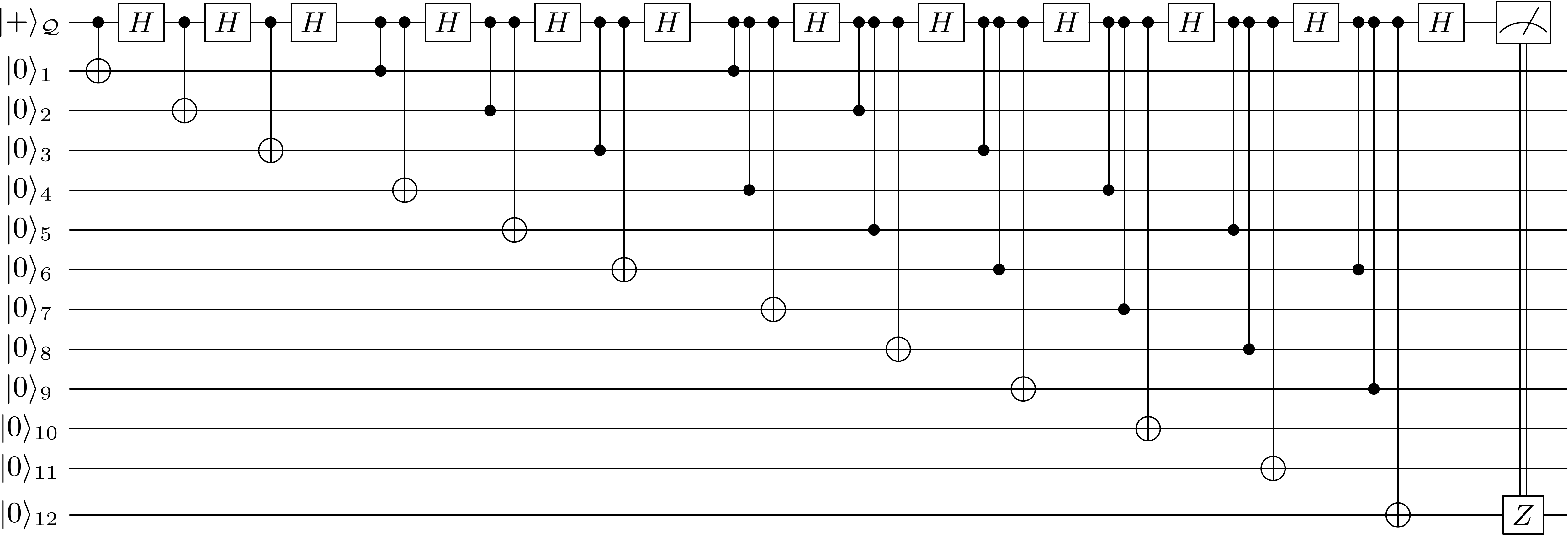}
\caption{\label{fig:circ_example}Algorithm~\ref{alg1} applied to the cubic lattice $\Gc$ [cf.~Eq.~\eqref{cubic edges}] for $L=3$, $M=2$, and $n=2LM$. This circuit prepares $\ket{\psi_\Gc}$ on the data qubits, and is the first step of Protocol~\ref{protA}. The last operation could be replaced by $Z_{\Q,n}$.}
\end{figure*}

Since the bcc lattice $\Gbcc$ is a subgraph of $\Gc$, we can then measure the qubits of $\ket{\psi_{\Gc}}$ that are not in $\Gbcc$ in the $Z$-basis to remove them. We also need to measure all of the qubits on the boundary in $\Gc$ in the $Z$-basis, in order to get rid of the shifted periodic boundary conditions.\footnote{Therefore, to prepare the cluster state on a $L \times M \times N$ bcc lattice using Protocol~\ref{protA}, we would generate the cluster state on a $(L + 2) \times (M + 2) \times (N+ 2)$ cubic lattice in Line~\ref{protA: 1}.} After applying the appropriate Pauli corrections based on the outcomes of these measurements, we obtain the desired cluster state $\ket{\psi_\Gbcc}$. 

\begin{protocol}[H] \caption{prepare the cluster state $\ket{\psi_{\Gbcc}}$ on the bcc lattice $\Gbcc = (V_{\mathrm{bcc}}, E_{\mathrm{bcc}})$ of Ref.~\cite{Raussendorf2006}} \label{protA}
\begin{algorithmic}[1]
\State apply Algorithm~\ref{alg1} to $\Gc = (V_{\mathrm{c}}, E_{\mathrm{c}})$ [cf.~Eq.~\eqref{cubic edges}] \label{protA: 1}
\For{$i \in V_{\mathrm{c}} \setminus V_{\mathrm{bcc}}$} \label{protA: for}
\State measure qubit $i$ in the $Z$-basis  \label{prota: measure} 
\If{the outcome is $\ket{1}$}
\State apply $\prod\limits_{j:(i,j) \in E_{\mathrm{c}}}Z_j$ \label{protA: correct} \label{protA: last}
\EndIf
\EndFor
\end{algorithmic}
\end{protocol}

\subsubsection{Protocol B} \label{sec:protB}
Protocol~\ref{protB} prepares $\ket{\psi_\Gbcc}$ by directly applying Algorithm~\ref{alg1} to $\Gbcc$.

\begin{protocol}[H] \caption{prepare the cluster state $\ket{\psi_{\Gbcc}}$ on the bcc lattice $\Gbcc = (V_{\mathrm{bcc}}, E_{\mathrm{bcc}})$ of Ref.~\cite{Raussendorf2006}} \label{protB}
\begin{algorithmic}[1]
\State apply Algorithm~\ref{alg1} to $\Gbcc$
\end{algorithmic}
\end{protocol}

For notational convenience in Section~\ref{sec:error}, we adopt the following convention for the bcc lattice. We label the qubits of an $L \times M \times N$ bcc lattice as we would an $L\times M\times N$ cubic lattice, omitting the numbers corresponding to the cubic lattice sites that are ``missing"---see Fig.~\ref{fig:bcc_labels} for an example illustrating this convention. This is a slight departure from the notation in Algorithm~\ref{alg1} (which assumes that the qubits are numbered from $1$ through $n$), but the instructions of Algorithm~\ref{alg1} can be adapted straightforwardly. Fig.~\ref{fig: protB circuit} shows part of the resulting circuit for the lattice in Fig.~\ref{fig:bcc_labels}. 

Using our labelling convention, the nearest neighbours of a qubit $i$ are simply $\{i \pm 1$, $i \pm L$, and/or $i \pm LM\} \cap V_{\mathrm{bcc}}$. Each of the qubits, except those on the boundary, has four nearest neighbours, all of which lie in the same plane. Thus, we divide the qubits into three groups, $V^{xy}$, $V^{yz}$, and 
$V^{zx}$, where qubit $i$ is in $V^{xy}$ (resp.\ $V^{yz}, V^{zx}$) if the nearest neighbours of $i$ are in the $xy$ (resp.\ $yz$, $zx$) plane. Letting $N_{\mathrm{bcc}}(i)$ denote the set of nearest neighbours of $i$ in $\Gbcc$, we have [cf.~Fig.~\ref{fig:bcc_labels}]
\begin{equation} \label{Nbcc}
N_{\mathrm{bcc}}(i) \subseteq \begin{dcases}
\{i \pm 1, i \pm L\} \quad &\text{$i \in V^{xy}$} \\
\{i \pm L, i\pm LM\} \quad &\text{$i \in V^{yz}$} \\
\{i \pm LM, i\pm 1\} \quad &\text{$i \in V^{zx}$}.
\end{dcases}
\end{equation}
For qubits $i$ that are in the bulk of the lattice, the above holds with equality.


\begin{figure}[H]
\begin{center}
\includegraphics[width=0.8\linewidth]{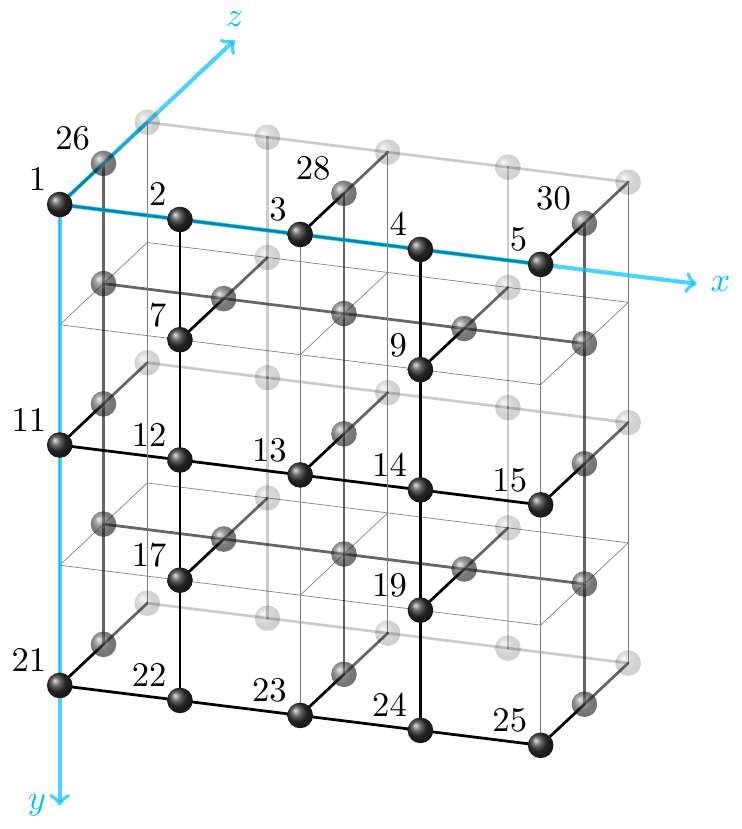}
\end{center}
\caption{An example of a bcc lattice $\Gbcc$, with $L = M = 5$, and an ordering of its vertices. As discussed in the main text, certain labels are skipped, so that the neighbours of vertex $i$ are given by $\{i \pm 1, i \pm L, i \pm LM\} \cap V_{\mathrm{bcc}}$. Labels after $30$ have been omitted. \label{fig:bcc_labels}}
\end{figure}

\begin{figure*}
\includegraphics[width=0.9\textwidth]{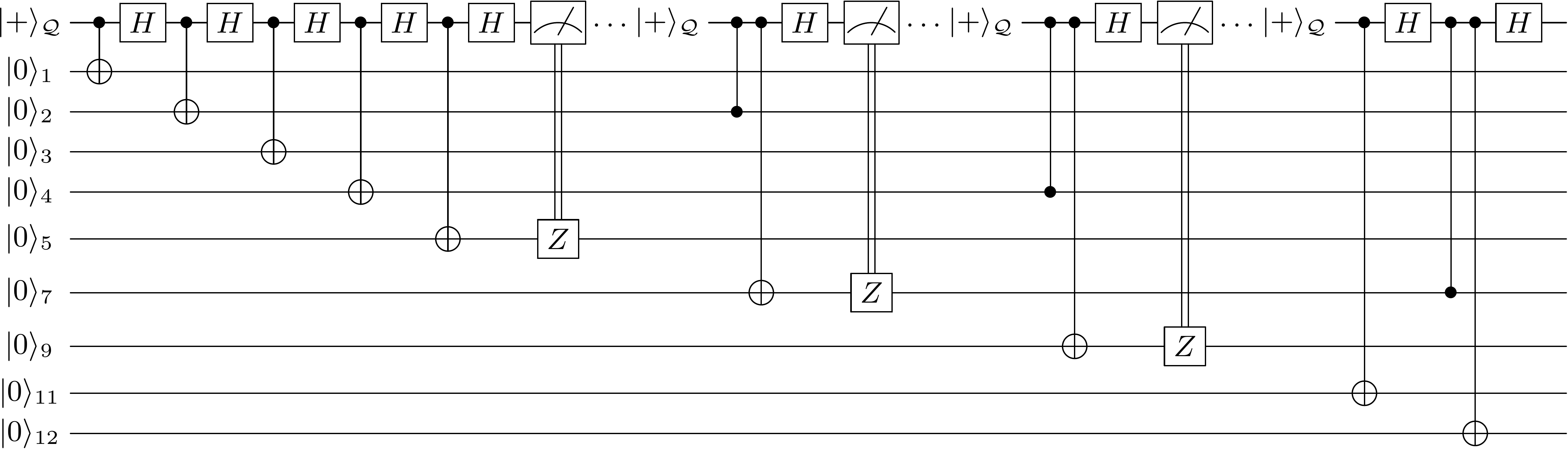}
\caption{\label{fig: protB circuit}Part of Protocol~\ref{protB} for the bcc lattice in Fig.~\ref{fig:bcc_labels}. Only the first $10$ iterations of the \textbf{for} loop in Algorithm~\ref{alg1} are shown. This abridged circuit prepares $\ket{\psi_{\Gbcc[12]'}}$ [cf.~Eq.~\eqref{G[k]'}]. Note that all Pauli $Z$ corrections conditioned on the outcomes of the measurements of $\Q$ can be deferred to the end of the circuit.} 
\end{figure*}

\section{Error analysis} \label{sec:error}

The protocols described in Section~\ref{sec:cluster} are useful only insofar as they are fault-tolerant. The operations used in the protocols will generally be noisy, resulting in the preparation of imperfect cluster states. Since the ancilla qubit $\Q$ interacts with every data qubit in Algorithm~\ref{alg1}, single-qubit errors occurring during the procedure may propagate through the subsequent operations to highly nonlocal errors. We show, however, that the effect of these errors on the target cluster state is always equivalent to that of geometrically local errors. This allows us to demonstrate that for both Protocols~\ref{protA} and~\ref{protB}, there is a threshold for the circuit error rate below which the logical error rate rapidly decays with the system size.

To make our reasoning precise, let $g_1, \dots, g_D$ denote the sequence of Clifford gates in Algorithm~\ref{alg1},\footnote{Lines~\ref{alg1: measure}--\ref{alg1: reset} of Algorithm~\ref{alg1} have the same combined effect as that of a deterministic Clifford gate, and can be treated as such for this discussion. A Pauli error occurring between Lines~\ref{alg1: measure} and \ref{alg1: reset} is equivalent to a Pauli error occurring after these steps.} and for $j,k \in [D]$, let $C_{j}^{k} \coloneqq \prod_{i = j}^{k} g_i$. Then, if a Pauli error $P_a$ occurs on some qubit $a$ between the gates $g_{\ell-1}$ and $g_\ell$, the erroneous circuit implements $C_\ell^D P_a C_1^{\ell-1}$. The prepared state is
\begin{equation} \label{propagateQ}
    C_\ell^D P_a C_1^{\ell - 1}\ket{\phi_{\textrm{initial}}} = Q\ket{\phi_{\textrm{final}}},
\end{equation}
where $\ket{\phi_{\textrm{initial}}}$ is the input state, $\ket{\phi_{\textrm{final}}} = C_1^D\ket{\phi_{\textrm{initial}}}$ denotes the state prepared by the error-free circuit, and $Q \coloneqq C_\ell^D P_a(C_\ell^D)^\dagger$. In other words, the circuit-level error $P_a$ propagates to an error $Q$, which may be highly nonlocal in general. In fact, for certain choices of $P_a$ and $\ell$, the weight of $Q$ scales with the total number of qubits.

However, Eq.~\eqref{propagateQ} holds for arbitrary $\ket{\phi_{\textrm{initial}}}$, with $Q$ independent of the initial state. The fact that errors propagate nonlocally for \textit{generic} input states is not necessarily an issue---the purpose of Algorithm~\ref{alg1} is not to perform some computation on arbitrary inputs, but rather, to prepare a fixed resource state. Therefore, the only relevant analysis is that for the particular input state $\ket{\phi_\textrm{initial}} \coloneqq \ket{+}_\Q \bigotimes_{i=1}^n\ket{0}_i$ to Algorithm~\ref{alg1}, which leads to the particular output state $\ket{\phi_{\textrm{final}}} = \ket{+}_\Q \ket{\psi_G}$. Clearly, $Q\ket{\phi_{\textrm{final}}} = QS\ket{\phi_{\textrm{final}}}$ for any stabiliser $S$ of $\ket{\phi_\textrm{final}}$. Therefore, even if $Q$ is a high-weight operator, it may have the same effect on $\ket{\phi_{\textrm{final}}}$ as a low-weight operator. 

It will hence be useful to define the notion of \textit{effective errors}. We say that a circuit-level Pauli error $P_a$ occurring at depth $\ell$ results in an effective error $E$ if
\begin{equation} \label{effective error} C_\ell^D P_a C_1^{\ell-1}\ket{\phi_{\textrm{initial}}} = E\ket{\phi_{\textrm{final}}}. \end{equation}
This definition generalises straightforwardly to arbitrary circuit-level errors. Note that unlike Eq.~\eqref{propagateQ}, Eq.~\eqref{effective error} is \textit{not} a gate identity, as it may depend crucially on the input state $\ket{\phi_{\textrm{initial}}}$. Note also that $E$ is not unique.

If multiple Pauli errors occur in the circuit, their joint effect is multiplicative up to a sign. To see this, consider two Pauli errors $P_a$ and $P_b$ occurring at depths $\ell_1$ and $\ell_2$, respectively, with $\ell_1 \leq \ell_2$. Suppose that the circuit-level error $P_a$ (at depth $\ell_1$) results in an effective error $E_1$, in the sense of Eq.~\eqref{effective error}, and $P_b$ (at depth $\ell_2$) results in an effective error $E_2$. Then, the circuit containing both errors prepares
\begin{align}
    &C_{\ell_2}^D P_b C_{\ell_1}^{\ell_2 - 1}P_a C_{1}^{\ell_1 - 1}\ket{\phi_{\textrm{initial}}} \nonumber\\ &\qquad= \left[C_{\ell_2}^D P_b (C_{\ell_2}^D)^\dagger\right] C_{\ell_1}^D P_a C_1^{\ell_1 - 1}\ket{\phi_{\textrm{initial}}} \nonumber\\
    &\qquad= \left[C_{\ell_2}^D P_b (C_{\ell_2}^D)^\dagger\right] E_1 C_1^D\ket{\phi_{\textrm{initial}}}  \nonumber\\
    &\qquad= (-1)^s E_1 C_{\ell_2}^D P_b C_1^{\ell_2 -1}\ket{\phi_{\textrm{initial}}}  \nonumber\\
    &\qquad= (-1)^s E_1E_2 \ket{\phi_{\textrm{final}}},
\label{multiple errors}
\end{align}
where the second and fourth equalities use Eq.~\eqref{effective error}, and the phase $(-1)^s$ is either $+1$ or $-1$ depending on whether $E_1$ and $C_{\ell_2}^D P_b (C_{\ell_2}^D)^\dagger$ (which are both Pauli products) commute or anticommute. Thus, the two circuit-level errors collectively result in an effective error $E_1 E_2$, up to a sign. Analogous results hold for more than two errors.

It follows that in order to study a stochastic noise model involving Pauli errors, it suffices to analyse the effective errors resulting from single-qubit circuit-level errors. The effect of multi-qubit circuit-level errors can then be inferred from Eq.~\eqref{multiple errors}.

As we discuss in Section~\ref{sec:4A}, any single-qubit error occurring during Protocols~\ref{protA} or~\ref{protB} results in a local effective error on the final state. This is a special case of the more general result, proven in Appendix~\ref{appendix:errors1}, for Algorithm~\ref{alg1} applied to arbitrary graphs. In Section~\ref{sec:threshold_result}, we estimate the threshold circuit error rates for both protocols, obtaining $0.23\%$ for Protocol~\ref{protA} and $0.39\%$ for Protocol~\ref{protB}.

\subsection{Effective errors} \label{sec:4A}

In this subsection, we consider the effect of errors that occur during Protocols~\ref{protA} and~\ref{protB}, both of which prepare the cluster state $\ket{\psi_{\Gbcc}}$ on the bcc lattice $\Gbcc$. These protocols both apply Algorithm~\ref{alg1} (but to different graphs). In Appendix~\ref{appendix:errors1}, we prove that for any graph $G = (V,E)$, any single-qubit error occurring between the elementary operations of Algorithm~\ref{alg1} results in an effective error [\textit{cf.}~Eq.~\eqref{effective error}] that is geometrically local, in the sense that it is supported within $\{i\} \cup N(i)$ for some data qubit $i \in [n]$, where $N(i) \coloneqq \{j:(i,j) \in E\}$ denotes the nearest neighbours of $i$ in $G$. 

The proof uses the following key observations. 
\begin{enumerate}[1)]
\item First, it is clear from Figs.~\ref{fig:circ_example} and~\ref{fig: protB circuit} that any $Z_i$ error on a data qubit $i \in [n]$ either occurs before the $X_{\Q,i}$ gate and has no effect, as the initial state of $i$ is $\ket{0}$, or it occurs after the $X_{\Q,i}$, in which case it commutes with all subsequent operations and ends up as a $Z_i$ error on the final state. Thus, any single-qubit $Z$ error on a data qubit results in either no error or a $Z$ error on the same qubit. 
\item Second, the instantaneous state of the qubits at any point in the procedure is a cluster state,\footnote{up to a Hadamard $H_{\Q}$ on the ancilla $\Q$} as illustrated by Fig.~\ref{fig:G[k]'}(b). In the underlying graph of any of these intermediate cluster states, every edge between data qubits is also an edge in the graph $G$ of the target state $\ket{\psi_G}$, and the only edges involving the ancilla $\Q$ are between $\Q$ and $j \in S$ for a subset $S$ of $N(i)$ for some $i \in [n]$. It then follows from the stabiliser condition, Eq.~\eqref{stabilisers}, that any single-qubit $X$ error in the circuit has the same effect as a set of $Z$ errors confined to the neighbours of some data qubit (and possibly $\Q$). 
\end{enumerate}
Combining these two observations with Eq.~\eqref{multiple errors}, it can be shown that any single-qubit Pauli error leads to an effective error of the form $\prod_{j \in S'} Z_j$, where $S' \subset \{i\} \cup N(i)$ for some $i \in [n]$. We fill in the details in Appendix~\ref{appendix:errors1}. Here, we simply summarise the results that are relevant to the threshold calculations in the following subsection.

We start by considering the effective errors in Protocol~\ref{protA}. Recall that the first step (Line~\ref{protA: 1}) applies Algorithm~\ref{alg1} to the cubic lattice $\Gc$ defined by Eq.~\eqref{cubic edges}, yielding a circuit of the form of Fig.~\ref{fig:circ_example}.
Table~\ref{table:errors_cubic} lists all of the $X$ and $Z$ errors that may occur in this circuit and the effective errors they give rise to. Note that it suffices to consider the effect of single-qubit $X$ and $Z$ errors, as the effect of arbitrary errors can then be inferred by decomposing them in terms of Pauli operators and using Eq.~\eqref{multiple errors}. To clearly distinguish between the gates in the circuit, we use $A_j$ to denote the $j$th ``block" of gates [\textit{cf.}~Line~\ref{alg1: Bj} of Algorithm~\ref{alg1}],\footnote{Here and in Table~\ref{table:errors_cubic}, we have chosen to apply $Z_{\Q,j-LM}$ before $Z_{\Q,j-L}$. The order of these controlled-$Z$ gates could of course be changed, in which case the second and third entries in Table~\ref{table:errors_cubic} would be slightly different (see Table~\ref{table:errors1} in Appendix~\ref{appendix:errors1}).}
\begin{equation} \label{Aj}
A_j \coloneqq H_\Q X_{\Q,j}Z_{\Q,j-L}Z_{\Q,j-LM}.
\end{equation}
Spatially, circuit-level errors may occur on the ancilla $\Q$ or a data qubit $i \in [n]$, and temporally, they may be located between two blocks $A_j$ and $A_{j+1}$, before the first block $A_1$, after the last block $A_n$, or between two gates in the same block. Table~\ref{table:errors_cubic} covers all of these possibilities. 

\begin{table*}
  \begin{tabular}{ c | l | l }
    circuit-level error &location in circuit &effective error on final state ($\ket{\psi_\Gc}$) \\
    \hhline{=|=|=}
    \multirow{5}{*}{$X_\Q$} 
    &before $A_1$ &none \\ 
    &immediately after $Z_{\Q,k-LM}$ in $A_k$  & $Z_{k-LM}Z_{k-1}$  \\ 
    &immediately after $Z_{\Q,k-L}$ in $A_k$ &$Z_{k-LM}Z_{k-L}Z_{k-1}$ \\ 
    &immediately after $X_{\Q,k}$ in $A_k$ &$Z_{k+1}$ \\
    &immediately after $B_k$ (\textit{i.e.}, after $H_\Q$ in $A_k$)  & $Z_k$ \\ \hline
    $Z_\Q$ &before or within $A_k$ &$Z_k$ \\ \hline
    \multirow{3}{*}{$X_i$} 
    &before (the $Z_{Q,i}$ in) $A_{i+L}$ & $X_{i}Z_{i+L}Z_{i+LM}$ \\ 
    &after $A_{i+L}$ and before $A_{i+LM}$ &$X_iZ_{i+LM}$ \\
    &after $A_{i+LM}$ &$X_i$ \\
    \hline 
    \multirow{2}{*}{$Z_i$} 
    &{before $X_{\Q,i}$ (in $A_i$)} &none\\
    &{after $X_{\Q,i}$ (in $A_i$)} &$Z_i$\\
\end{tabular}\caption{\label{table:errors_cubic}A complete list of $X$ and $Z$ errors that could occur during Line~\ref{protA: 1} of Protocol~\ref{protA}, which prepares a cluster state on the cubic lattice $\Gc$, and their effect on the final state [\textit{cf.}~Eq.~\eqref{effective error}] up to a sign. $A_j$ is defined in Eq.~\eqref{Aj}. We use the convention that $Z_j \equiv I$ for any $j \notin [n]$. Additionally, any $Z_i$ error occurring during Lines~\ref{protA: for}--\ref{protA: last} results in a $Z_i$ error, while an $X_i$ error results in either $X_i$ or $\prod_{j \in S} Z_j$ for some $S \subseteq N_{\mathrm{c}}(i)$, depending on its precise location.} 
\end{table*}

By Eq.~\eqref{cubic edges}, the set $N_\mathrm{c}(i)$ of nearest neighbours of qubit $i$ in $\Gc$ is \begin{equation*} \label{Nc} N_\mathrm{c}(i) = \{i \pm 1, i \pm L, i \pm M\} \cap [n].\end{equation*} Hence, we can see from Table~\ref{table:errors_cubic} (and Eq.~\eqref{stabilisers}) that any single-qubit $X$ error results in an effective error of the form $\prod_{j \in S} Z_j$ up to a sign, where $S \subset N_\mathrm{c}(i)$ for some $i \in [n]$, while any $Z$ error results in either no effective error or a single-qubit $Z$ error. Moreover, $X$ and $Z$ errors occurring at the same spacetime location in the circuit result in effective ($Z$) errors supported within
$\{i\} \cup N_{\mathrm{c}}(i)$ \emph{for the same $i$}, which implies that \emph{any} single-qubit error occurring at that location leads to an effective error supported within $\{i\} \cup N_{\mathrm{c}}(i)$. This is easily verified using Table~\ref{table:errors_cubic}.
As an example, an $X$ error on $\Q$ between gates $Z_{\Q,k-LM}$ and $Z_{\Q,k-L}$ in $A_k$ results in an effective error $Z_{k-LM}Z_{k-1}$, while a $Z$ error at this location results in a $Z_k$ error, and $k, k-LM, k-1 \in \{k\} \cup N_{\mathrm{c}}(k)$.

Therefore, at the end of Line~\ref{protA: 1} of Protocol~\ref{protA: 1}, the effective error induced by any single-qubit error can be decomposed into $Z$ operators supported within some neighbourhood of $\Gc$. Note that the remaining steps, Lines~\ref{protA: for}--\ref{protA: last}, of Protocol~\ref{protA} do not propagate this effective error further, as $Z$ errors commute with $Z$ gates and do not affect $Z$-measurements. By the same argument, a single-qubit $Z$ error occurring during Lines~\ref{protA: for}--\ref{protA: last} does not propagate to other qubits. It is also clear that a single-qubit $X$ error on qubit $i$ occurring during these steps is equivalent to $Z$ errors on a subset of $N_\mathrm{c}(i)$. It follows that the effective error on $\ket{\psi_\Gbcc}$ resulting from any single-qubit error in Protocol~\ref{protA} is geometrically local with respect to $\Gc$, \textit{i.e.,} supported within the neighbourhood $\{i\} \cup N_\mathrm{c}(i)$ of some qubit $i \in [n]$. Such an error is also geometrically local with respect to $\Gbcc$ if $i \in V_{\mathrm{bcc}}$, while if $i \not\in V_{\mathrm{bcc}}$, it is still confined to an elementary cell of $\Gbcc$ [\textit{cf.}~Fig.~\ref{fig:bcc}].

An even stronger result holds for Protocol~\ref{protB}, which directly prepares $\ket{\psi_\Gbcc}$ using Algorithm~\ref{alg1}. Table~\ref{table:errors_bcc} lists the effective errors resulting from all possible single-qubit $X$ and $Z$ errors. In the table, $B_j$ denotes the block of gates applied in the \textbf{for} loop iteration of Algorithm~\ref{alg1} (for $G = \Gbcc$) corresponding to qubit $j$. Recalling the labelling convention for $\Gbcc$ described in Section~\ref{sec:protB},
\begin{equation} \label{Bj}
    B_j \coloneqq \begin{dcases}
    H_\Q X_{\Q,j} Z_{\Q,j-L} &\quad j \in V^{xy} \\
    H_\Q X_{\Q,j} Z_{\Q,j-L}Z_{\Q,j-LM} &\quad j \in V^{yz} \\
    H_\Q X_{\Q,j}Z_{\Q,j-LM} &\quad j \in V^{zx}
    \end{dcases}
\end{equation}
As shown in Fig.~\ref{fig: protB circuit}, $\Q$ is measured and reset between certain gate blocks, and
Table~\ref{table:errors_bcc} includes the effects of measurement and reset errors as well. It is clear from Table~\ref{table:errors_bcc} and Eq.~\eqref{Nbcc} that the effective error induced by any single-qubit Pauli error is equivalent to a product of $Z$ operators supported within $\{i\} \cup N_{\mathrm{bcc}}(i)$ for some $i \in V_{\mathrm{bcc}}$. Thus, single-qubit errors occurring at any spacetime location in Protocol~\ref{protB} result in effective errors on $\ket{\psi_{\Gbcc}}$ that are geometrically local with respect to $\Gbcc$. 

Since all vertices in $\Gc$ and $\Gbcc$ have constant degree, it follows (from Eq.~\eqref{multiple errors}) that any $m$-qubit circuit-level error results in an effective error of weight $cm$ for some constant $c$ independent of the system size. Standard arguments then imply that for both Protocols~\ref{protA} and~\ref{protB}, there is a finite threshold for the circuit error rate~\cite{Dennis2002,Bravyi2020}. We compute these thresholds in the following subsection.

\begin{table*}
\begin{center}
\begin{tabular}{c | l | c | c | c}
circuit-level error &location in circuit &\multicolumn{3}{c}{effective error on final state ($\ket{\psi_{\Gbcc}}$)} \\
\hhline{=|=|=|=|=}
\multirow{6}{*}{$X_\Q$} & &$k \in V^{xy}$ &$k \in V^{yz}$ &$k \in V^{zx}$ \\ \cline{3-5} 
&immediately after $Z_{\Q,k-LM}$ in $B_k$ &n/a &$Z_{k-LM}$ &$Z_{k-LM}Z_{k-1}$ \\
&immediately after $Z_{\Q,k-L}$ in $B_k$ &$Z_{k-L}Z_{k-1}$ &$Z_{k-LM}Z_{k-L}$ &n/a \\
&immediately after $X_{\Q,k}$ in $B_k$ &$Z_{k+1}$ &none & $Z_{k+1}$ \\
&immediately after $B_k$ (\textit{i.e.}, after $H_\Q$ in $B_k$) &$Z_k$ &$Z_k$ &$Z_k$ \\ \cline{2-5}
&before $B_1$ or after a re-initialisation of $\Q$ &\multicolumn{3}{c}{none} \\ \hline
\multirow{2}{*}{$Z_\Q$} &before or within $B_k$ &\multicolumn{3}{c}{$Z_k$} \\
&before a $Z$-measurement of $\Q$ &\multicolumn{3}{c}{none} \\ \hline
\multirow{4}{*}{$X_i$} & &$i \in V^{xy}$ &$i \in V^{yz}$ &$i \in V^{zx}$ \\ \cline{3-5}
&before $B_{i+L}$ &$X_i Z_{i+L}$ &$X_i Z_{i+L}Z_{i+LM}$ &$X_i Z_{i+LM}$ \\
&after $B_{i+L}$ and before $B_{i+LM}$ &$X_i$ &$X_i Z_{i+LM}$ &$X_i Z_{i+LM}$ \\
&after $B_{i+LM}$ &$X_i$ &$X_i$ &$X_i$ \\ \hline
\multirow{2}{*}{$Z_i$} &before $X_{\Q,i}$ (in $B_i$) &\multicolumn{3}{c}{none} \\
&after $X_{\Q,i}$ (in $B_i$) &\multicolumn{3}{c}{$Z_i$}
\end{tabular}
\end{center}
\caption{A complete lists of $X$ and $Z$ errors that could occur during Protocol~\ref{protB}, which prepares a cluster state on the bcc lattice $\Gbcc$, and their effect on the final state [\textit{cf.}~Eq.~\eqref{effective error}] up to a sign. The qubits are labelled according to the convention described in Section~\ref{sec:protB}, and $B_j$ is defined in Eq.~\eqref{Bj}. $Z_j \equiv I$ for any indices $j$ that are out of range. The results for $X_\Q$ errors occurring immediately after $X_{\Q,k}$ (row 3) hold for all qubits $k$ that are in the bulk of the bcc lattice. However, as can be seen from Fig.~\ref{fig:bcc_labels}, there are certain qubits $k \in V^{xy}, V^{zx}$ on the boundary for which $(k,k+1) \not\in E$. In these cases, $\Q$ is measured after $B_k$ (Line~\ref{alg1: measure} of Algorithm~\ref{alg1}) and there is no effective error. \label{table:errors_bcc}}
\end{table*} 

We make a side remark on the role of intermediate measurements. It is tempting to guess that these measurements are responsible for the locality of the effective errors, but that is emphatically not the case. In Protocol~\ref{protA}, no intermediate measurements are ever performed during the preparation of $\ket{\psi_\Gc}$, yet all of the effective errors are geometrically local with respect to $\Gc$ [cf.~Table~\ref{table:errors_cubic}]. It is surprising that there is a nontrivial \textit{extensive-depth} fault-tolerant protocol without intermediate measurements; the usual approach involves frequent intermediate measurements to extract syndrome information, so that one can catch the errors. In contrast, we only perform error correction at the very end, after an extensive-depth circuit has been executed. Finding necessary and sufficient conditions under which this is possible is an important open problem left for future work.

\subsection{Thresholds}
\label{sec:threshold_result}

Using the results of the previous subsection, we can calculate error thresholds for our protocols via Monte Carlo simulations. 
In order to compare Protocols~\ref{protA} and~\ref{protB} to the standard cluster state preparation circuit in Ref.~\cite{Raussendorf2006}, we consider the standard depolarising model (Error Model 1 below) and use the minimum-weight perfect matching (MWPM) decoder \cite{Raussendorf2006,edmonds_1965}. We also study the effect of qubit loss (Error Model 2) using the decoder of Ref.~\cite{Barrett2010}, which is also based on MWPM.

For various values of the circuit error rate $p$, loss error rate $p_{\textrm{loss}}$, and size $L$, we estimate the logical error rate $\overline{p}$ for generating an $L\times L\times L$ cluster state $\ket{\psi_{\Gbcc}}$ (storing one logical qubit). We average over at least $10^6$ independent instances and at least $10^4$ logical errors for each set of parameters. For each $p_{\textrm{loss}}$, we then estimate the threshold circuit error rate $p_{\textrm{th}}$ by fitting the data to a quadratic scaling ansatz
\begin{equation} \label{scaling} \overline{p} = a + b(p-p_{\textrm{th}})d^{1/\nu} + c(p-p_{\textrm{th}})^2d^{2/\nu}, \end{equation}
where $d = (L+1)/2$.

\subsubsection{Error Model 1}

Error Model 1 is the standard depolarising model. In this model, every single-qubit gate on a qubit $a$ is followed by a single-qubit depolarising channel
\begin{equation} \label{depolarising1} \mathcal{D}_a^{(p)}(\rho) = (1-p)\rho + \frac{p}{3}\sum_{P \in \{X,Y,Z\}}P_a\rho P_a \end{equation}
on $a$. In addition, every (re-)initialisation of $a$ is followed by $\mathcal{D}_a^{(p)}$, and every measurement of $a$ is preceded by $\mathcal{D}_a^{(p)}$. Here, measurements include not only those in Protocols~\ref{protA} and~\ref{protB}, but also the eventual $X$-measurements on data qubits that are required for extracting the syndrome.
Similarly, every two-qubit gate on qubits $a$ and $b$ is followed by a two-qubit depolarising channel 
\begin{equation} \label{depolarising2}  \mathcal{D}_{a,b}^{(p)}(\rho) = (1-p)\rho + \frac{p}{15}\sum_{\substack{P,P' \in \{I,X,Y,Z\}\\ (P,P') \neq (I,I)}} P_a P'_b \rho P_a P_b'. \end{equation} 
We refer to $p$ as the circuit error rate.

For Protocols~\ref{protA} and~\ref{protB}, we can simulate the effect of each of these depolarising errors on the final state using Tables~\ref{table:errors_cubic} and~\ref{table:errors_bcc}. Our results (along with the fits to Eq.~\eqref{scaling}) are plotted in Fig.~\ref{fig:thresholds}. The threshold circuit error rate $p_{\textrm{th}}$ is found to be $0.23\%$ for Protocol~\ref{protA} and $0.39\%$ for Protocol~\ref{protB}.

\begin{figure}
    \centering
    \subfloat[]{
    \includegraphics[clip, trim=0 1em 0 4em, width=0.95\columnwidth]{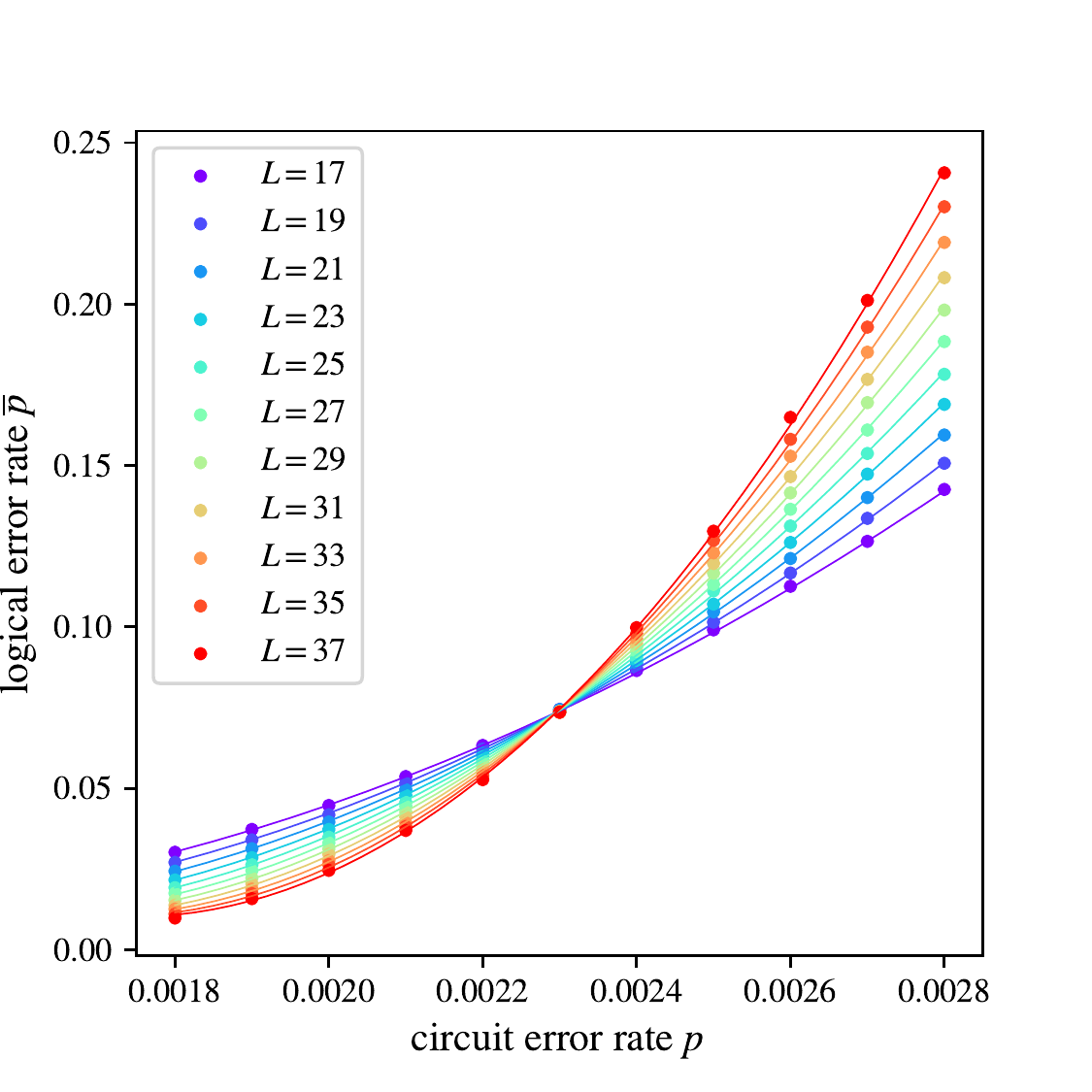}
    }
    
    \subfloat[]{
    \includegraphics[clip, trim=0 1em 0 4em, width=0.95\columnwidth]{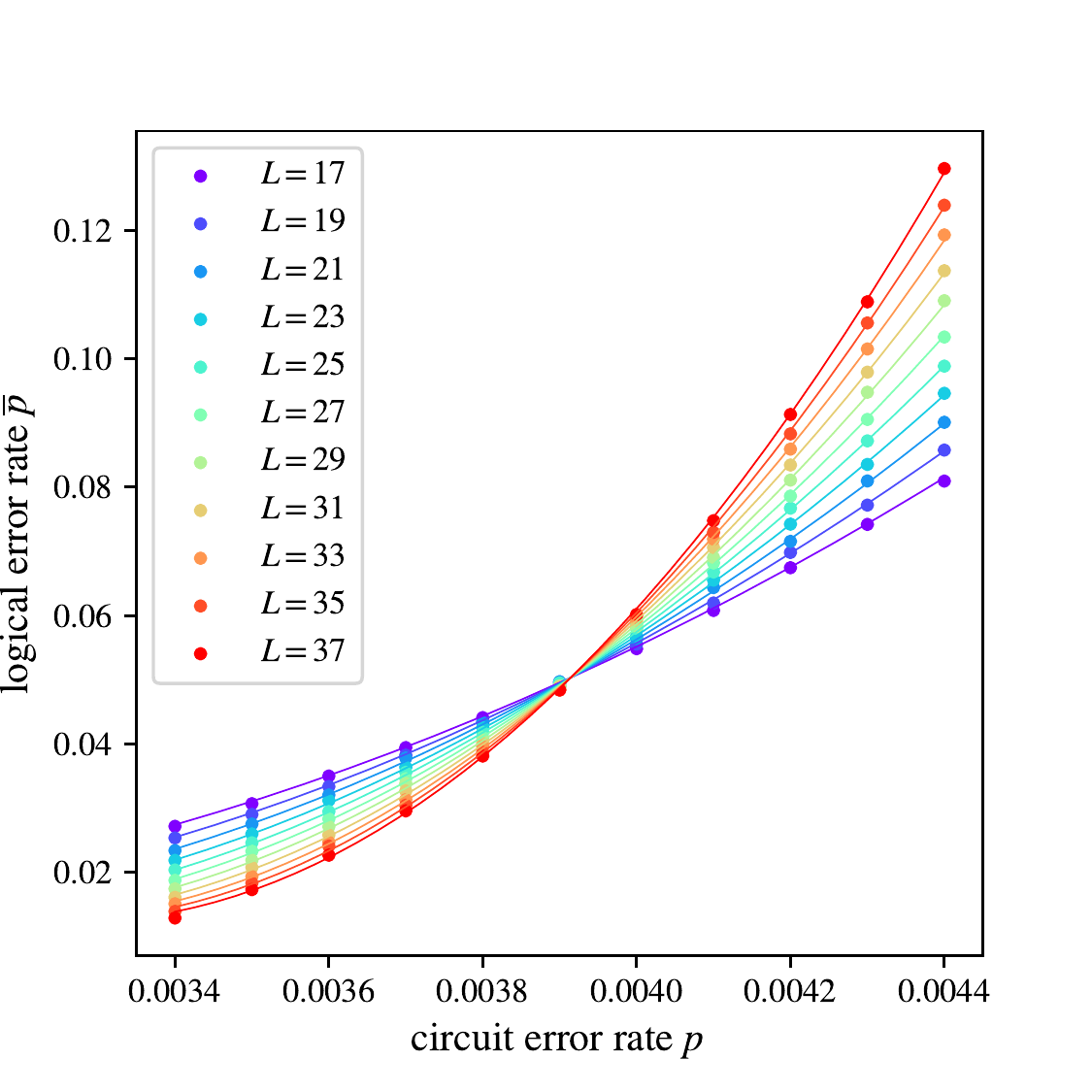}
    }
    \caption{Logical error rate $\overline{p}$ vs. circuit error rate $p$ for Error Model~1, for (a)~Protocol~\ref{protA} and (b)~Protocol~\ref{protB}. Solid curves are fits to Eq.~\eqref{scaling}, which give $p_{\textrm{th}} = 0.23\%$ for Protocol~\ref{protA} and $p_{\textrm{th}} = 0.39\%$ for Protocol~\ref{protB}. \label{fig:thresholds}}
\end{figure}

In comparison, the threshold for the scheme of Ref.~\cite{Raussendorf2006} under the same error model is $0.58\%$. Refs.~\cite{Raussendorf2007,Raussendorf2007a}  improve this to $0.75\%$ by exploiting sublattice correlations, and Ref.~\cite{Barrett2010} obtains $0.63\%$ by accounting for the degeneracies of different matchings. We do not exploit correlations nor account for degeneracy in our decoder.

We surmise that the threshold for Protocol~\ref{protA} is lower than that for Protocol~\ref{protB} due to the following reasons. First, Protocol~\ref{protA} uses substantially more qubits and operations than Protocol~\ref{protB} to prepare a cluster state of the same size, giving rise to more error locations under Error Model~1.
Second, all of the effective errors in Protocol~\ref{protB} are geometrically local with respect to the bcc lattice $\Gbcc$, whereas some of the effective errors in Protocol~\ref{protA} are only geometrically local with respect to the cubic lattice $\Gc$. For example, suppose that an $X$ error occurs on a qubit $i \in V_{\mathrm{c}} \setminus V_{\mathrm{bcc}}$ immediately before the $Z$-measurement of $i$ in Line~\ref{prota: measure} of Protocol~\ref{protA}. By Eq.~\eqref{stabilisers}, this results in a $Z$ error on all of the neighbours of $i$ in $G_{\mathrm{c}}$, which constitutes a weight-$6$ error on the face qubits of an elementary cell of $\Gbcc$ [\textit{cf.}~Fig.~\ref{fig:bcc}]. In contrast, all of the effective errors resulting from single-qubit errors in Protocol~\ref{protB} are geometrically local with respect to $\Gbcc$, and, moreover, have weight at most $4$ (when restricted to either the primal or dual lattice). 

\subsubsection{Error Model 2}

Next, we add detectable loss errors to the standard depolarising noise model. In Error Model~2, every elementary operation is followed or preceded by a depolarising channel with error rate $p$ in exactly the same way as in Error Model~1. In addition, each data qubit is lost by the end of the procedure with probability $p_{\textrm{loss}}$. Hence, Error Model~2 reduces to Error Model~1 for $p_{\textrm{loss}} = 0$. We assume that if a qubit $i$ is lost at some point, then any subsequent operation on $i$ is replaced by the identity operator followed by depolarising noise with rate $p$. The assumption that losses are detectable and that operations involving lost qubits implement the identity is consistent with the experimental setup considered in Section~\ref{sec:experiment}. 

Fig.~\ref{fig:loss_thresholds} shows our estimates for the threshold circuit error rate $p_{\textrm{th}}$ at various values of $p_{\textrm{loss}}$. The solid line in each plot is a quadratic fit to the data. Extrapolating to $p_{\textrm{th}} = 0$, these fits give rough estimates for the loss threshold of $5.7\%$ for Protocol~\ref{protA} and $21.6\%$ for Protocol~\ref{protB}. 
Both plots have the same structure as Figure~3 in Ref.~\cite{Barrett2010}, which provides thresholds for the circuit of Ref.~\cite{Raussendorf2006} under the same error model (but using a slightly better decoder, as discussed above). 

\begin{figure}
    \centering
    \subfloat[]{
    \includegraphics[clip, trim=1.5em 1em 0 1em, width=0.85\columnwidth]{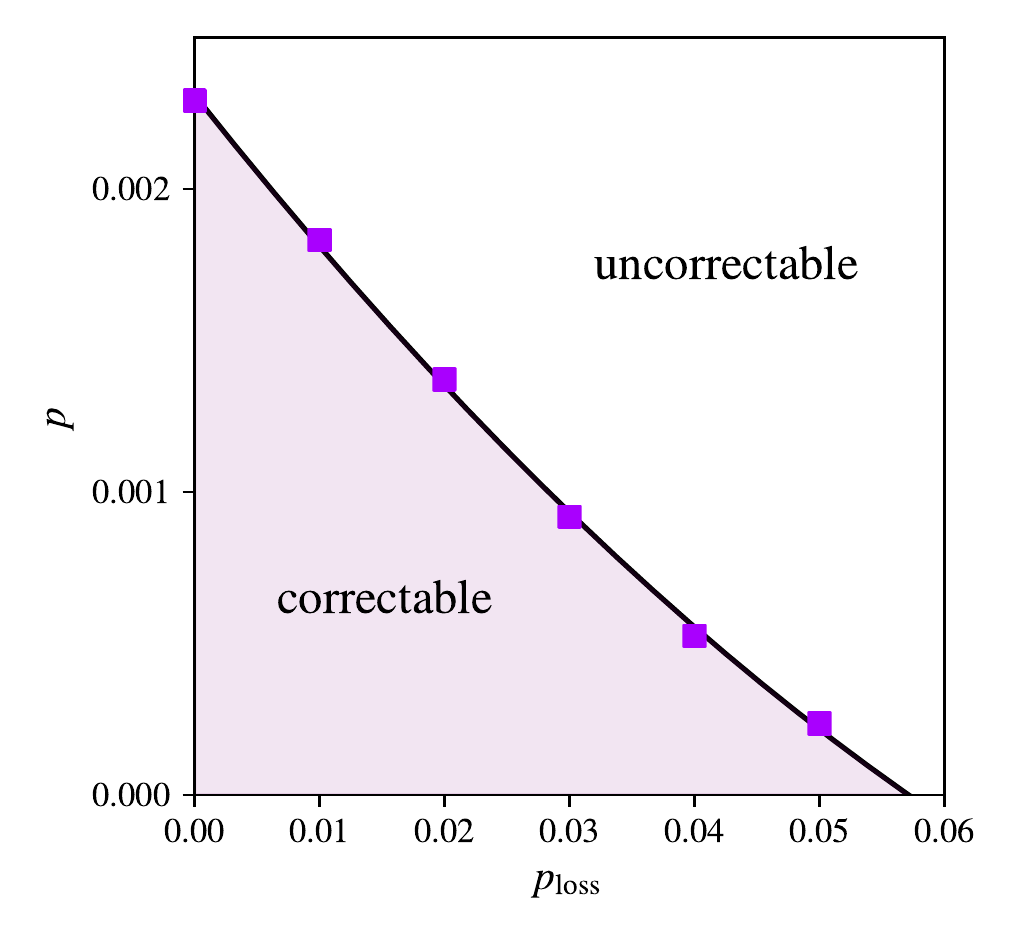}
    }

    \subfloat[]{
    \includegraphics[clip, trim=1.5em 1em 0 1em, width=0.85\columnwidth]{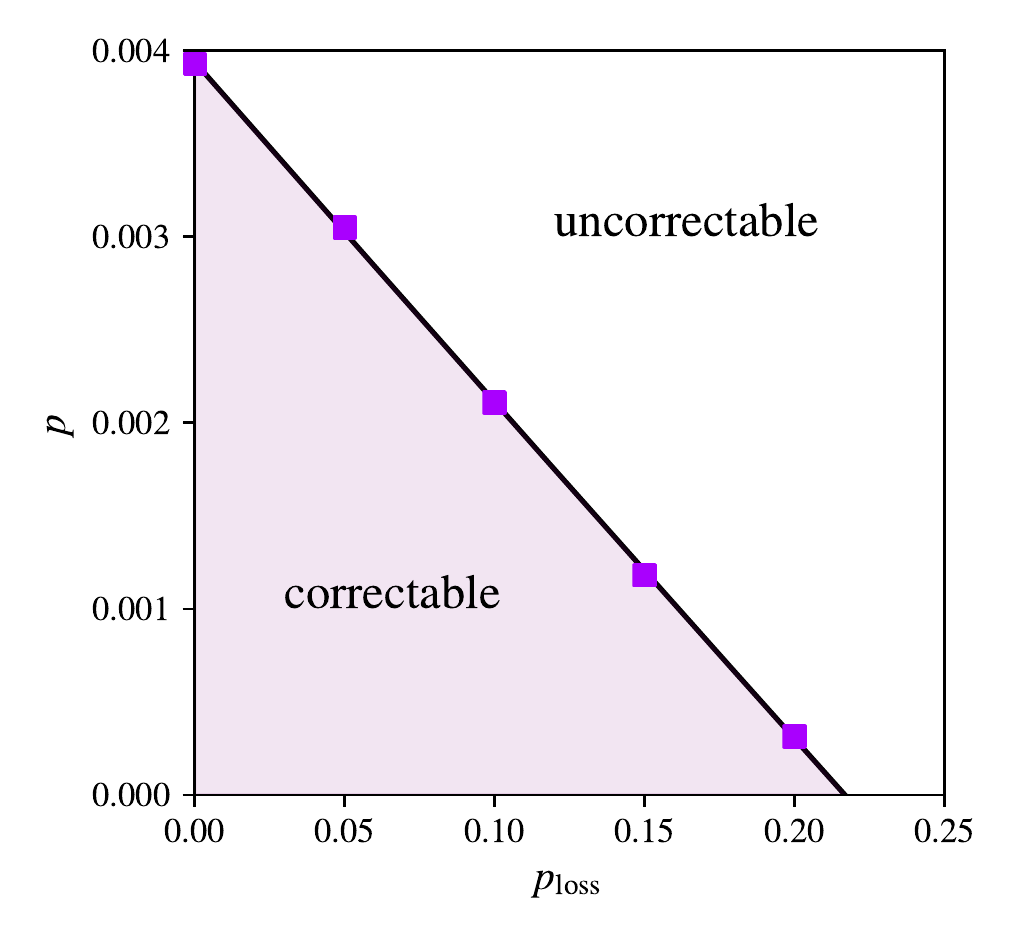}
    }
    \caption{Squares indicate thresholds $p_\textrm{th}$ for the circuit error rate $p$ at various loss rates $p_{\text{loss}}$ in Error Model~2, for (a)~Protocol~\ref{protA} and (b)~Protocol~\ref{protB}. Solid curves are quadratic fits to the data. The shaded region ($p < p_{\textrm{th}}$) represents the correctable region of parameter space, in which the logical error rate can be suppressed by increasing the system size. \label{fig:loss_thresholds}}
\end{figure}

The loss threshold for Protocol~\ref{protA} is significantly lower than that for Protocol~\ref{protB} due to the fact that in our simulations, losing a qubit in $V_{\mathrm{c}} \setminus V_{\mathrm{bcc}}$ amounts to losing (up to) six qubits in $V_{\mathrm{bcc}}$. This is because if a qubit $i \in V_{\mathrm{c}} \setminus V_{\mathrm{bcc}}$ is lost, we would not know whether the correction $\prod_{j \in N_\mathrm{c}(i)}Z_j$ should be applied in Line~\ref{protA: last} of Protocol~\ref{protA}. Instead of simulating this as a weight-$6$ $Z$ error (with probability $1/2$), we simply treat all of the qubits in $N_{\mathrm{c}}(i)$ as having been lost in the decoding algorithm. Thus, the total probability of ``losing" a qubit in $V_{\mathrm{bcc}}$ is greater than $p_{\textrm{loss}}$ for Protocol~\ref{protA}.

While Error Model~2 allows for a direct comparison to Ref.~\cite{Barrett2010}, and may be an informative model for settings where the total loss probability is constant, it does not properly capture the structure of the noise expected when storing the data qubits in delay lines. Informed by the description of possible experimental implementations in the next section, we revisit the effect of delay line noise in Section~\ref{sec:delay_line}.

\section{Experimental realisation} \label{sec:experiment}

In this section, we outline potential experimental realisations of the abstract protocols in Section~\ref{sec:prep_3d}, focusing on implementations in quantum nanophotonic and acoustic systems~\cite{RevModPhys.87.347,Schoelkopf2017quantum}. Recent advances in the deterministic generation of single photons and single phonons~\cite{PhysRevX.4.041010,PhysRevLett.116.020401} and their coherent interactions with a single quantum emitter~\cite{gustafsson2014propagating,volz2014nonlinear,sipahigil2016integrated,goban2014atom,tiecke2014nanophotonic,reiserer2014quantum,PhysRevLett.121.040501,maccabe2019phononic} make these systems promising platforms for quantum information processing. Indeed, single and double chains of one-dimensional cluster states have already been produced in experiments using photons emitted from quantum dots~\cite{schwartz2016deterministic}. These experiments implement modified versions of the circuit in Ref.~\cite{Lindner2009}, which is a specific instance of Algorithm~\ref{alg1}. The techniques detailed in Refs.~\cite{Pichler2017,Lindner2009} can be adapted to our more general protocols, to create cluster states on different graphs. In particular, the experimental setup considered in Ref.~\cite{Pichler2017} can be directly extended to implement the first step of Protocol~\ref{protA} [\textit{cf.}~Fig.~\ref{fig:circ_example}], providing a simple procedure for preparing a three-dimensional cluster state on a cubic lattice. Universal fault-tolerant quantum computation can then be performed by making adaptive single-qubit measurements on this state~\cite{Raussendorf2006}.

There are several key ingredients required for realising Protocols~\ref{protA} and~\ref{protB}. First, we must be able to implement the elementary operations in these protocols, namely, the single-qubit operations on $\Q$, the controlled-$X$ gates $X_{\Q,i}$ and controlled-$Z$ gates $Z_{\Q,i}$ between $\Q$ and data qubits, and single-qubit measurements of the data qubits. Second, we must be able to coordinate the interactions between $\Q$ and the data qubits such that these operations are applied in the correct order. Additionally, to be able to perform error correction when the loss rate is significant, the qubit states must be encoded in such a way that losses are detectable. 

These capabilities can be naturally achieved in a system consisting of a single quantum emitter (\textit{e.g.}, an atom, ion, transmon, or quantum dot) coupled to a photonic or phononic waveguide [\textit{cf.}~Fig.~\ref{fig:exp_gates}(a)].
In such a system, any stable internal states of the emitter can be used to encode qubit degrees of freedom for $\Q$, while any radiative states of the emitter that are coupled to the waveguide can be leveraged to realise certain gates between $\Q$ and a photon or phonon propagating in the waveguide. We show below that the set of available gates is sufficient for Algorithm~\ref{alg1}. Moreover, the routing of the photons or phonons required to realise the geometry of the target graphs of Protocols~\ref{protA} and~\ref{protB} is rather simple.

\subsection{Encoding schemes and elementary gates} \label{sec:encodings_physical}

In this subsection, we describe two encoding schemes and the gates that can be implemented in each. We will refer to these as the \emph{single-rail} and the \emph{dual-rail} encoding schemes, summarised in Figs.~\ref{fig:exp_gates} and~\ref{fig:exp_gates_dual_rail}, respectively. 

\subsubsection{Single-rail encoding} \label{sec:single_rail}

\begin{figure*}
    \centering
    \includegraphics[width=0.8\textwidth]{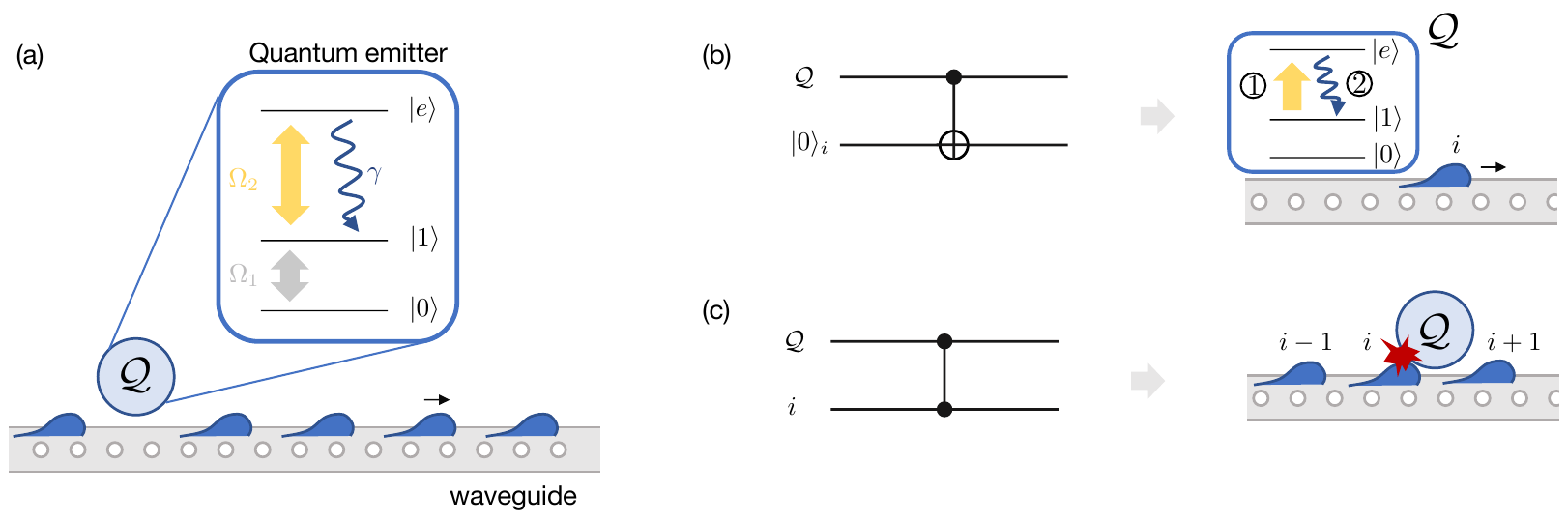}
    \caption{Schematic diagrams illustrating the implementation of required gates in quantum photonic or acoustic systems.
    (a)~The quantum emitter $\Q$ has three relevant quantum states. Two stable states, $\ket{0}$ and $\ket{1}$, form a qubit, while an extra unstable excited state $\ket{e}$ is used to generate propagating photons or phonons in a guided mode.
    (b)~The operation $\widetilde{X}_{\Q,i}$, which has the same effect as $X_{\Q,i}$ when acting on $\ket{0}_i$~[\textit{cf.}~Eq.~\eqref{cX_single}], can be implemented via selective emission of a photon/phonon to the guided mode.
    (c)~A controlled-$Z$ gate $Z_{\Q,i}$ can be implemented via the scattering of a photon/phonon against the emitter.}
    \label{fig:exp_gates}
\end{figure*}

In the single-rail scheme, the $\ket{0}$ (resp.\ $\ket{1}$) state of each data qubit $i$ is encoded by the absence (resp.\ presence) of a photon or phonon. Multiple data qubits can be encoded in a single waveguide by controlling the rate of the excitation pulses, in the so-called time multiplexing technique.
More specifically, if the pulse-to-pulse time separation $\tau$ is sufficiently long compared to the temporal extent of an emitted photon/phonon mode, different modes separated by $\tau$ have exponentially small overlap~\cite{Pichler2017}. We note that the temporal extent of each emitted mode, or equivalently the effective emission rate $\gamma'$, can be controlled using advanced techniques such as pulse shaping~\cite{Pichler2017,PhysRevX.4.041010}. 

For the emitter $\Q$, we consider a three-level system consisting of two stable states, $\ket{0}$ and $\ket{1}$, along with a radiative state, $\ket{e}$ [\textit{cf.}~Fig.~\ref{fig:exp_gates}(a)]. Arbitrary single-qubit gates on $\Q$ can be realised via resonant coherent excitations between $\ket{0}$ and $\ket{1}$. 

Then, for the two-qubit gates $X_{Q,i}$, note that each $X_{\Q,i}$ in Algorithm~\ref{alg1} is applied when data qubit $i$ is in its initial state $\ket{0}$. This means that instead of implementing a controlled-$X$ gate $X_{\Q,i}$ that correctly transforms arbitrary states of $i$, we can use an operation $\widetilde{X}_{\Q,i}$ that has the same effect as $X_{\Q,i}$ when $i$ is in the specific state $\ket{0}$ (\textit{i.e.}, $\widetilde{X}_{\Q,i}\ket{\varphi}_\Q\ket{0}_i = X_{\Q,i}\ket{\varphi}_\Q\ket{0}_i$ for any state $\ket{\varphi}$ of $\Q$, potentially entangled with the rest of the system). $\widetilde{X}_{\Q,i}$ can be realised by applying a rapid resonant excitation pulse $\ket{1}\rightarrow \ket{e}$, which is followed by the spontaneous emission of a photon/phonon into the waveguide [cf.~Fig.~\ref{fig:exp_gates}(b)]. This excitation-emission process deterministically generates a single photon/phonon in a particular temporal mode (controlled by the timing of the excitation pulse and the decay rate $\gamma$), conditioned on the state of $\Q$ being $\ket{1}$. Thus, the net effect is
\begin{equation} \label{cX_single}
\begin{aligned}
    &\ket{0}_{\Q}\ket{0}_i \ket{\phi_0}_\textrm{rest} + \ket{1}_{\Q}\ket{0}_i \ket{\phi_1}_\textrm{rest} \\ 
    &\mapsto
    \ket{0}_{\Q}\ket{0}_i \ket{\phi_0}_\textrm{rest} + \ket{1}_{\Q}\ket{1}_i \ket{\phi_1}_\textrm{rest},
\end{aligned}
\end{equation}
where $\ket{\phi_0}_\textrm{rest}$ and $\ket{\phi_1}_\textrm{rest}$ are (unnormalised) states of the rest of the system.

The controlled-$Z$ gate $Z_{\Q,i}$ can be naturally realised by scattering a propagating photon/phonon against the emitter $\Q$~\cite{volz2014nonlinear,sipahigil2016integrated,goban2014atom,tiecke2014nanophotonic,reiserer2014quantum} [cf.~Fig.~\ref{fig:exp_gates}(c)].
If $\Q$ is in the state $\ket{0}$, the propagating photon/phonon remains unaffected due to the absence of any resonant couplings. On the other hand, if $\Q$ is in the state $\ket{1}$, the propagating photon (phonon) is scattered by $\Q$ owing to the resonant transition $\ket{1} \leftrightarrow \ket{e}$, giving rise to a scattering phase $e^{i\theta}$.  By engineering $\gamma' \ll \gamma$, this scattering phase approaches $e^{i\theta} \approx -1$, and this process effectively applies $Z_{\Q,i}$.

\subsubsection{Dual-rail encoding} \label{sec:dual_rail}

\begin{figure*}
    \centering
    \includegraphics[width=0.8\textwidth]{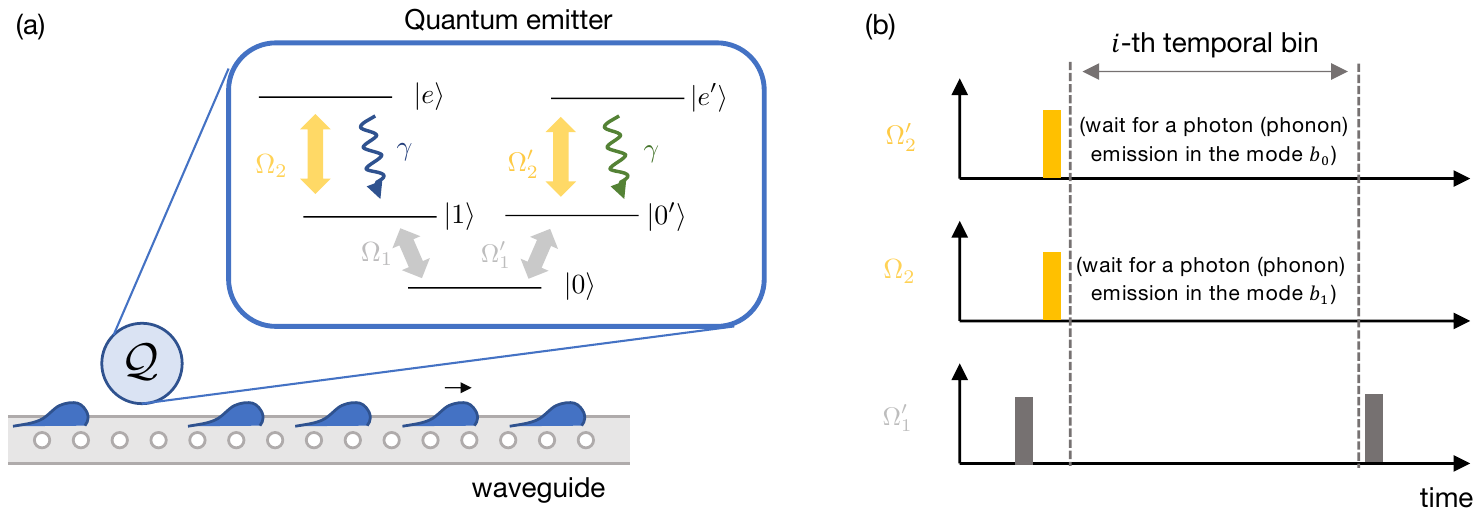}
    \caption{Schematic diagrams illustrating how to encode quantum information in two distinct modes of a single photon state (dual-rail schemes).
    (a) Two extra quantum states, $\ket{0'}$ and $\ket{e'}$, of the emitter allows to map the $\ket{0}$ into a distinct photonic (phononic) mode in the waveguide.
    (b) An example of the pulse sequence that effectively realises the $X_{\Q,i}$ gate acting on a quantum state with the $i$th data qubit in $\ket{0}$. [cf.~Eq.~\eqref{cX_dual}]
    All yellow and grey boxes represent resonant $\pi$-pulses. 
    }
    \label{fig:exp_gates_dual_rail}
\end{figure*}

In the dual-rail scheme, a qubit degree of freedom is encoded in two distinct internal modes of a single photon/phonon, such as different polarisations or frequencies. When photon/phonon loss is the dominant source of error, the dual-rail scheme can be advantageous since the detection of a single photon/photon heralds the absence of loss errors (assuming no false-positive detections). As shown in Section~\ref{sec:threshold_result}, the threshold for loss errors is significantly higher than that for depolarising noise, for both Protocols~\ref{protA} and~\ref{protB}.

The gate implementations proposed for the single-rail scheme can be readily extended to the dual-rail scheme. For example, we can use two additional internal quantum states $\ket{0'}$ and $\ket{e'}$ of the emitter $\Q$ [cf.~Fig.~\ref{fig:exp_gates_dual_rail}(a)]. Similar to the states $\ket{1}$ and $\ket{e}$, we assume that $\ket{0'}$ and $\ket{e'}$ are stable and radiative, respectively. In particular, $\ket{e'}$ rapidly decays into $\ket{0'}$ by emitting a photon/phonon into the waveguide. In general, the photons/phonons emitted from $\ket{e}$ and $\ket{e'}$ are distinguishable by their internal modes.
We denote these modes using two distinct annihilation operators, $b_{1}$ and $b_{0}$.

Then, the realisation of the ${X}_{\Q,i}$ gate (more precisely, the preparation of the state $X_{\Q,i}\ket{\varphi}_\Q\ket{0}_i$, for arbitrary $\ket{\varphi}$) in the dual-rail scheme can be achieved via a sequence of resonant $\pi$-pulses between the $\ket{0} \leftrightarrow \ket{0'}$, $\ket{1} \leftrightarrow \ket{e}$, and $\ket{0'} \leftrightarrow \ket{e'}$ transitions [cf.~Fig.~\ref{fig:exp_gates_dual_rail}(b)]. First, a rapid resonant excitation $\ket{0} \rightarrow \ket{0'}$ is applied, leading to the process
\begin{equation*}
\begin{aligned}
&\ket{0}_{\Q} \ket{\varnothing}_i \ket{\phi_0}_\textrm{rest} + \ket{1}_{\Q} \ket{\varnothing}_i \ket{\phi_1}_\textrm{rest}  \\
&\mapsto 
\ket{0'}_{\Q} \ket{\varnothing}_i \ket{\phi_0}_\textrm{rest} + \ket{1}_{\Q} \ket{\varnothing}_i \ket{\phi_1}_\textrm{rest} ,
\end{aligned}
\end{equation*}
where $\ket{\varnothing}_i$ is the vacuum initial state of the $i$th temporal bin in the waveguide and $\ket{\phi_{0}}_\textrm{rest}$ and $\ket{\phi_{1}}_\textrm{rest}$ are unnormalised states of the rest of the system. Second, resonant excitation pulses are applied to both the $\ket{1}\rightarrow \ket{e}$ and $\ket{0'} \rightarrow \ket{e'}$ transitions, which is followed by the emission of a photon/phonon at the $i$th bin in $b_1$ or $b_0$, depending on the internal state of the emitter.
The state of the system after this emission is
\begin{align*}
\ket{0'}_{\Q} \left( b^\dagger_{0} \ket{\varnothing}_i \right)  \ket{\phi_0}_\textrm{rest}
+ 
\ket{1}_{\Q} \left( b^\dagger_{1} \ket{\varnothing}_i\right)  \ket{\phi_1}_\textrm{rest}.
\end{align*}
Finally, another resonant $\pi$-pulse is used to move the population from $\ket{0'}$ to $\ket{0}$.
The net effect of these processes is the map
\begin{equation} \label{cX_dual}
\begin{aligned}
&\ket{0}_{\Q} \ket{\varnothing}_i \ket{\phi_0}_\textrm{rest} + \ket{1}_{\Q} \ket{\varnothing}_i \ket{\phi_1}_\textrm{rest} \\
&\mapsto 
\ket{0}_{\Q} \left( b^\dagger_{0} \ket{\varnothing}_i \right)  \ket{\phi_0}_\textrm{rest}
+ 
\ket{1}_{\Q} \left( b^\dagger_{1} \ket{\varnothing}_i\right)  \ket{\phi_1}_\textrm{rest}.
\end{aligned}
\end{equation}
Hence, by identifying $\ket{0}_i  \equiv b^\dagger_0 \ket{\varnothing}_i$ and $\ket{1}_i  \equiv b^\dagger_1 \ket{\varnothing}_i$, these operations achieve the desired effect.

The realisation of the controlled-$Z$ gate remains unmodified from Section~\ref{sec:single_rail}. That is, $Z_{\Q,i}$ can be implemented via a simple resonant photon/phonon scattering process, since a photon/phonon in the mode $b_0$ does not interact with the $\ket{0}$ nor $\ket{1}$ states of the emitter.

\subsection{Implementation details}

We now explain how to use the encoding schemes and elementary operations described in Section~\ref{sec:encodings_physical} to implement Protocols~\ref{protA} and~\ref{protB}. On top of being able to realise the required gates individually, we need to route the data qubits so that these gates are applied in the correct order. Moreover, for Protocol~\ref{protB}, we require the ability to perform intermediate measurements on the emitter $\Q$. We provide the details below.

For both protocols, we must control the ordering of the sequential interactions between $\Q$ and the photons/phonons representing the data qubits. This can be achieved by introducing time-delayed feedback. In the proposal of Ref.~\cite{Pichler2017}, a single delay line is used to generate a cluster state on a two-dimensional square lattice with shifted periodic boundary conditions. This procedure can be generalised to prepare the cluster state on the cubic lattice $\Gc$ defined in Section~\ref{sec:protA}, by introducing two delay lines of appropriate lengths to realise the circuit of Line~\ref{protA: 1} of Protocol~\ref{protA} [cf.~Fig.~\ref{fig:circ_example}].

Specifically, consider the setup illustrated in Fig.~\ref{fig:protocol_schematic}, which involves two routers and two delay lines.
The routers are configured such that a propagating photon/phonon travels through each delay line only once before being measured at the output port. By setting the delay times for Delay~$1$ and Delay~$2$ to $\tau_{D1} = (3L-1)\tau $ and $\tau_{D2} = (3L(M-1)-1)\tau$, respectively, where $\tau$ is the pulse-to-pulse time separation between distinct data qubit modes, one obtains a cluster state on an $L\times M \times N$ cubic lattice $\Gc$. The multiplicative factor of $3$ accounts for the fact that there are up to three interactions in every block [cf.~Eq.~\eqref{Aj}].
The delay times $\tau_{D1}$ and $\tau_{D2}$ can be effectively tuned using well-established techniques such as electromagnetically induced transparency~\cite{Fleischhauer2005EIT_review} for coherent atomic media or band-structure engineering for photonic/phononic crystals~\cite{baba2008photonic_crystal_review}.   
Fig.~\ref{fig:exp_dual_rail} shows the ``circuit diagram'' for $L = 3$, $M = N= 2$ implemented in a quantum photonic or acoustic system (compare to Fig.~\ref{fig:circ_example}), in the single-rail encoding scheme. The data qubits are initialised in vacuum $\ket{\varnothing} = \ket{0}$, and the ${X}_{\Q,i}$ gates in Fig.~\ref{fig:circ_example} are implemented via selective photon/phonon emissions, represented by $b^\dagger$. The diagram for the dual-rail scheme is analogous, with the controlled-$b^\dagger$ operations replaced by the process in Eq.~\eqref{cX_dual}. 
\begin{figure*}
    \centering
    \includegraphics[width=0.9\textwidth]{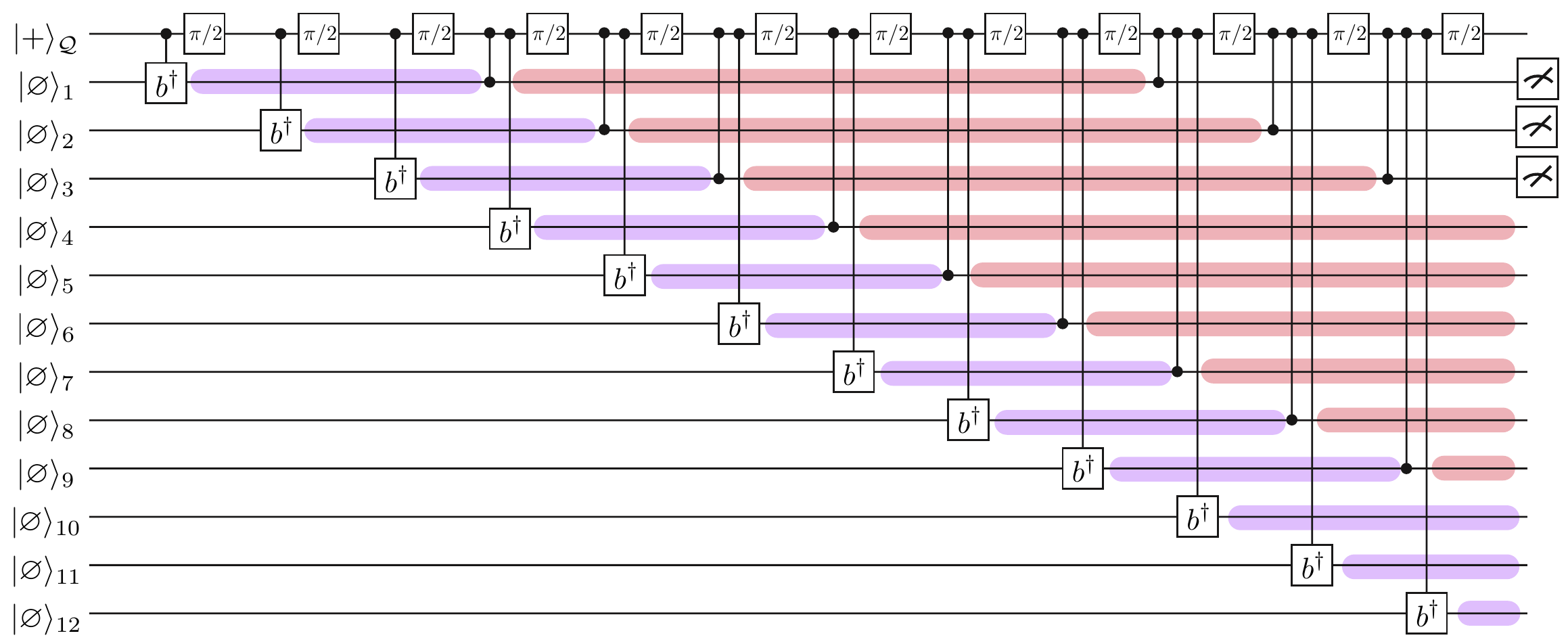}
    \caption{Explicit implementation of the circuit of Protocol~\ref{protA} for $L = M=3$
    in a quantum nanophotonic or acoustic system (compare to Fig.~\ref{fig:circ_example}).
    The quantum emitter $\Q$ sequentially creates and interacts with the data qubits, whose states in the single-rail scheme are encoded by the absence or presence of a photon/phonon in a guided, propagating mode.
    The generalisation of this procedure to the dual-rail scheme is straightforward.
    As detailed in Section~\ref{sec:encodings_physical}, the ${X}_{\Q,i}$ gates in Fig.~\ref{fig:circ_example} can be implemented via selective photon/phonon emission, while the $Z_{\Q,i}$ gates can be realised via resonant scattering.
    The colours correspond to those in Fig.~\ref{fig:protocol_schematic}, indicating which delay lines the photons/phonons are propagating through in the setup therein.}
    \label{fig:exp_dual_rail}
\end{figure*}

Protocol~\ref{protB} can be implemented in a similar way, with the following modifications. First, since the bcc lattice $\Gbcc$ is a subgraph of the cubic lattice $\Gc$, a subset of the cubic lattice sites should not host data qubits. Second, one needs to perform projective measurements on the emitter $\Q$, followed by re-initialisations of $\Q$ in a predetermined state, \textit{e.g.}, $\ket{+}$. The first task can be easily achieved by simply skipping the excitation pulses for the ${X}_{\Q,i}$ gates at the appropriate times. The data qubits at the corresponding locations then remain in the vacuum state, decoupled from the rest of the system throughout the procedure.

The measurements of $\Q$ can be implemented via quantum non-demolition measurements. A practical challenge is that the time duration $\tau_\textrm{meas}$ of the measurement and re-initialisation process can be substantial, constraining the minimum separation $\tau > \tau_\textrm{meas}/3$ between the temporal modes of consecutive data qubits. In turn, a longer $\tau$ implies that fewer data qubits can be stored in delay lines with non-negligible loss rates. It may therefore be more practical to use Protocol~\ref{protA} in some settings, to avoid intermediate measurements of $\Q$ altogether.\footnote{Further, for an alternative algorithm that prepares arbitrary cluster states without any measurements of $\Q$, see Appendix~\ref{appendix:alternative}.}

\section{Effect of delay line errors}
\label{sec:delay_line}

In the experimental proposals of Section~\ref{sec:experiment}, the amount of time a photon/phonon spends in the delay lines grows with the size of the target cluster state (more precisely, with the cross-sectional area of the underlying bcc lattice). As a result, the cumulative effect of delay line errors may be non-negligible. If the cluster state size becomes too large, the total error incurred in the delay lines overwhelms the improved error correction properties due to the increased code distance, leading to high logical error rates.
In this section, we study two phenomenological models that incorporate the effect of delay line errors, and determine how the optimal logical error rate depends on the delay line error rate in each model. These models address the increase in loss probability with delay line length as well as dephasing errors on the data qubits, two effects that were ignored in Section~\ref{sec:error}.

\subsection{Analysis} \label{sec:delay_lineA}

The dominant sources of error in delay lines (for photons and phonons) are dephasing and loss, which we consider in Error Models~3a and~3b, defined below. We parametrise these models using the delay line error \textit{per time step}. Here, the time it takes to execute the block of gates in Line~\ref{alg1: Bj} of Algorithm~\ref{alg1} constitutes one time step, and the time to measure and reset $\Q$ in Lines~\ref{alg1: measure} and~\ref{alg1: reset} constitutes another time step. We assume that the time steps are equal.\footnote{To elaborate, we assume that the operations are temporally arranged (by pulsing or measuring $\Q$ at appropriately spaced intervals) such that these time steps are all of the same length. This allows for the geometry of the bcc lattice to be naturally realised in Protocol~\ref{protB}, since the measurements of $\Q$ coincide with the sites in the bcc lattice that are ``missing" relative to the cubic lattice [cf.~Fig.~\ref{fig:bcc_labels}].} In Protocol~\ref{protB}, for instance, each time step consists either of a controlled-$X$ gate, a Hadamard gate, and up to two controlled-$Z$ gates, or of a measurement and re-initialisation of $\Q$. Thus, for preparing a cluster state on an $L\times M \times N$ lattice, there are $L$ time steps in Delay~1 and $L(M-1)$ time steps in Delay~2 [cf.~Figs.~\ref{fig:protocol_schematic} and~\ref{fig:exp_dual_rail}]. We will use $\eta_{Z}$ and $\eta_{\textrm{loss}}$ to represent the dephasing error and loss per time step, respectively.

Let us make a brief comment on how these delay line errors change the analysis in Section~\ref{sec:error}. Dephasing errors are generally equivalent to stochastic Pauli $Z$ errors, while the effect of losing a qubit during the procedure depends on the encoding scheme and gate implementations. In the experimental setup described in Section~\ref{sec:experiment}, the loss of a data qubit simply results in any subsequent gates involving that qubit not being applied. This is because these gates are realised via interactions between a photon/phonon wavepacket with the emitter. If the wavepacket is not present, the interaction does not occur. Moreover, in the dual-rail scheme of Section~\ref{sec:dual_rail}, losses are detectable. We can therefore use the decoder of Ref.~\cite{Barrett2010} as we did in Section~\ref{sec:threshold_result}. 

For circuit errors, we use the same depolarising noise model, Error Model~1, as in Section~\ref{sec:threshold_result}, but with one modification. We omit the depolarising channel that occurs after the initialisation of each data qubit.
This is motivated by the fact that in our experimental setup, each photon/phonon is created by the process that implements the controlled-$X$ gate [cf.~Eq.~\eqref{cX_dual}].
Strictly, this differs from Algorithm~\ref{alg1} (as it is formally stated), in which a data qubit $i$ is initialised first, before $X_{\Q,i}$ is applied to it in a separate step.
In order to be able to suppress error, the circuit error rate $p$ should be below the threshold for this modified error model. Since the threshold for Error Model~1 was estimated to be $p_{\textrm{th}} = 0.39\%$ (for Protocol~\ref{protB}), we will assume that $p = 10^{-3}$, which is a standard number used in the literature for studying the sub-threshold behavior of the surface code~\cite{Fowler2012}.

With the above considerations in mind, we define Error Models~3a and~3b as follows. We fix $p = 10^{-3}$. In both models, every single-qubit gate on a qubit $a$ is followed by a single-qubit depolarising channel $\mathcal{D}_a^{(p)}$ [cf.~Eq.~\eqref{depolarising1}]. Every measurement of a qubit $a$ is preceded by $\mathcal{D}_a^{(p)}$, and every (re-)initialisation of $\Q$ is followed by $\mathcal{D}_\Q^{(p)}$. Similarly, every two-qubit gate on qubits $a$ and $b$ is followed by a two-qubit depolarising channel $\mathcal{D}_{a,b}^{(p)}$ [cf.~Eq.~\eqref{depolarising2}]. In Error Model~3a, in addition to these circuit errors, a dephasing channel 
\[ \mathcal{Z}^{(\eta_Z)}(\rho) = (1-\eta_Z)\rho + \eta_Z Z\rho Z \]
is applied to every data qubit in each time step. In Error Model~3b, each data qubit is lost by the end of the procedure with probability $1 - \exp({-\eta_{\textrm{loss}}\ell})$, where $\ell$ is the total number of time steps the qubit spends in the delay lines. If a qubit $i$ is lost at some point, then any subsequent operation on $i$ is replaced by the identity operator followed by $\mathcal{D}_i^{(p)}$. In the following discussion, we will often refer to the delay line error rate as $\eta$, where $\eta$ is $\eta_Z$ for Error Model~3a and $\eta_\textrm{loss}$ for Error Model~3b.\footnote{In reality, one would expect both dephasing and loss errors to occur in an experiment. Analysing their effects separately greatly simplifies the numerics, however. In many circumstances, one form of noise will dominate, in which case the results for the corresponding error model should provide a good guide to the performance of the protocol.}

For each of these error models, we estimate the logical error rate $\overline{p}$ for generating a cluster state on an $L \times L \times L$ bcc lattice (storing one logical qubit) using Protocol~\ref{protB}, for various values of $L$ and $\eta$. As in Section~\ref{sec:threshold_result}, we infer the effect that each physical error has on the final state using Table~\ref{table:errors_bcc} and Eq.~\eqref{multiple errors}, and we use the generalised MWPM decoder of Ref.~\cite{Barrett2010} (without accounting for degeneracies) in our simulations. The results are shown in Fig.~\ref{fig:pertimestep1}. Each data point is an average of at least $10^6$ independent instances and at least $10^4$ logical errors.

\begin{figure}
    \centering
    \subfloat[]{
    \includegraphics[clip, trim=0 1em 0 4em, width=0.95\columnwidth]{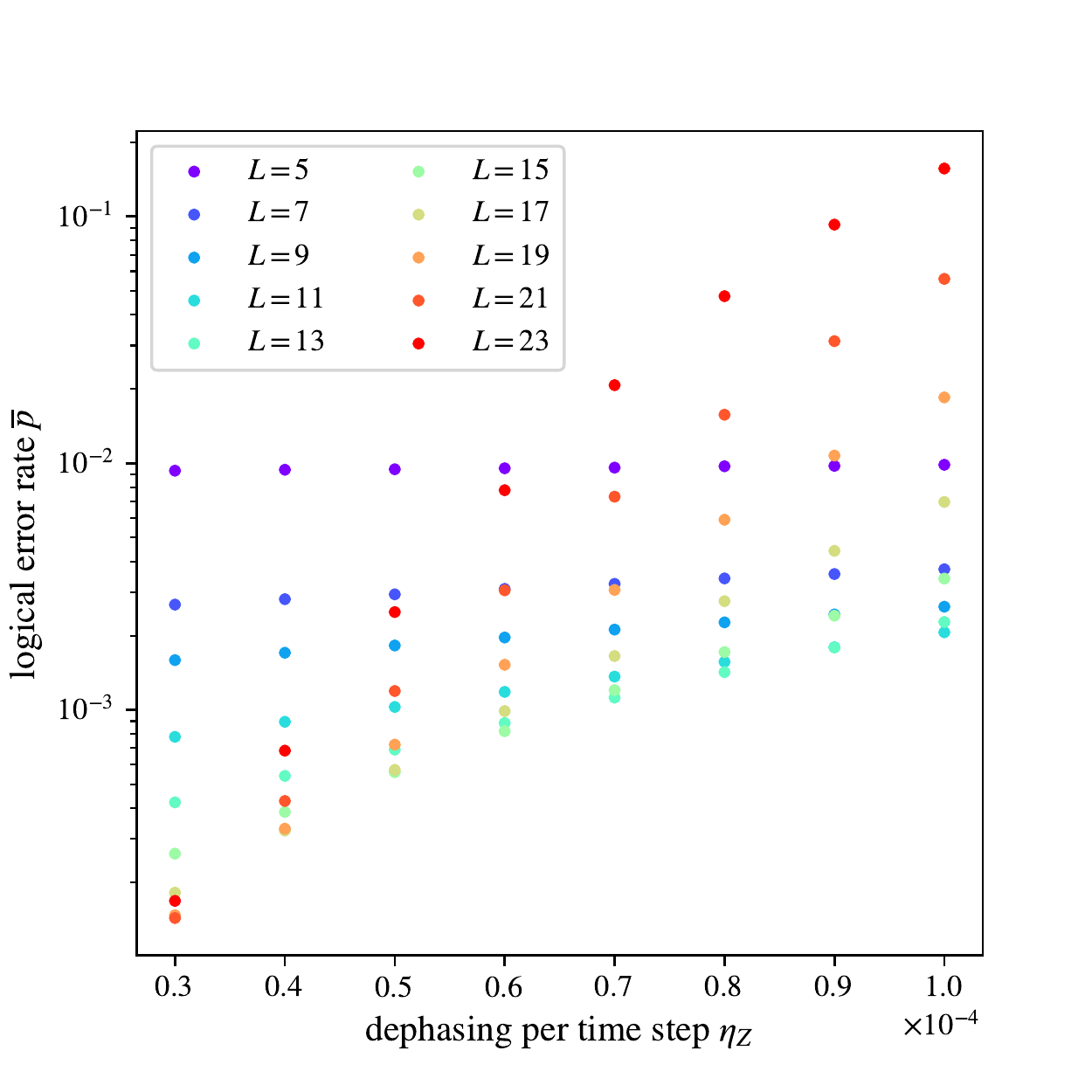}
    }
    
    \subfloat[]{
    \includegraphics[clip, trim=0 1em 0 4em, width=0.95\columnwidth]{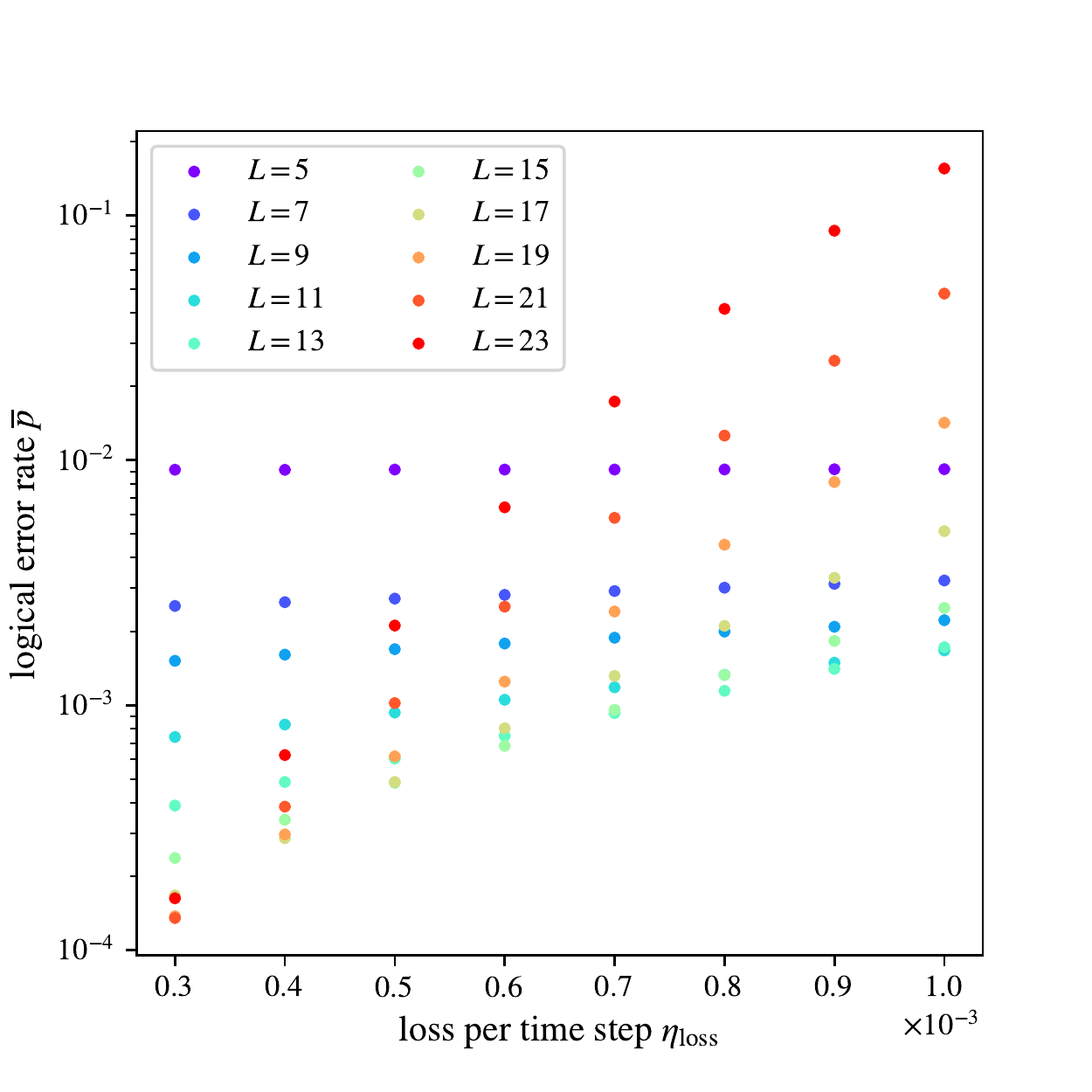}
    }
    \caption{(a) Logical error rate $\overline{p}$ vs. $\eta_{Z}$ for Error Model~3a (applied to Protocol~\ref{protB}). (b) Logical error rate $\overline{p}$ vs. $\eta_{\textrm{loss}}$ for Error Model~3b applied to Protocol~\ref{protB}. \label{fig:pertimestep1}}
\end{figure}

For an $L \times L \times L$ bcc lattice, the total number of delay line time steps is $L^2$. Hence, as we increase $L$ for a fixed delay line error rate $\eta$, there is a tradeoff between the better error suppression due to larger code distance, given by $d = (L+1)/2$, and the larger cumulative delay line error. Therefore, for each $\eta$, there is an optimal value $L_*$ of $L$ that minimises the logical error rate $\overline{p}$. We find the minimum logical error rate, which we denote by $\overline{p}_*$, by increasing $L$ until $\overline{p}$ starts to increase.

While we do not have an analytic expression for $L_*$, we can make an educated guess as to the scaling of $\overline{p}_*$ as a function of $\eta$. Since the circuit error rate $p = 10^{-3}$ in our models is well below threshold, there should be a threshold $p_{\textrm{delay},\textrm{th}}$ for the cumulative delay line error below which the logical error rate $\overline{p}$ decays exponentially with $L$.
For small $\eta$, the cumulative error is $\eta L^2$ to leading order, so $\overline{p}$ decays exponentially with $L$ for $\eta L^2 \lesssim p_{\textrm{delay},\textrm{th}}$. 
In particular, provided that 
$\eta L_*^2 \lesssim p_{\textrm{delay},\textrm{th}}$,
which can be achieved by 
$L_* = c(p_{\textrm{delay},\textrm{th}}/\eta)^{1/2}$ for some constant~$c$,
we expect $\overline{p}_*$ to roughly scale as
\begin{equation} \label{eq:logical_scaling_eta} \ln(1/\overline{p}_*) \approx c'\eta^{-1/2} + c'', \end{equation}
where $c'$ and $c''$ are constants. 

Numerically, we observe excellent agreement with Eq.~\eqref{eq:logical_scaling_eta}, as can be seen from Fig.~\ref{fig:optimal_vs_eta}. Fitting Eq.~\eqref{eq:logical_scaling_eta} to the data gives the following estimates for $c'$ and $c''$:
\begin{equation}
\begin{aligned}
&\text{Error Model~3a}: &\text{$c' \approx 0.032$, $c''\approx 2.93$} \\
&\text{Error Model~3b}: &\text{$c'\approx 0.096$, $c''\approx 3.37$}
\end{aligned}
\label{eq:extrapolation_constants}
\end{equation}

\begin{figure}
    \centering
    \subfloat[]{
    \includegraphics[clip, trim=0 1em 0 4em, width=0.95\columnwidth]{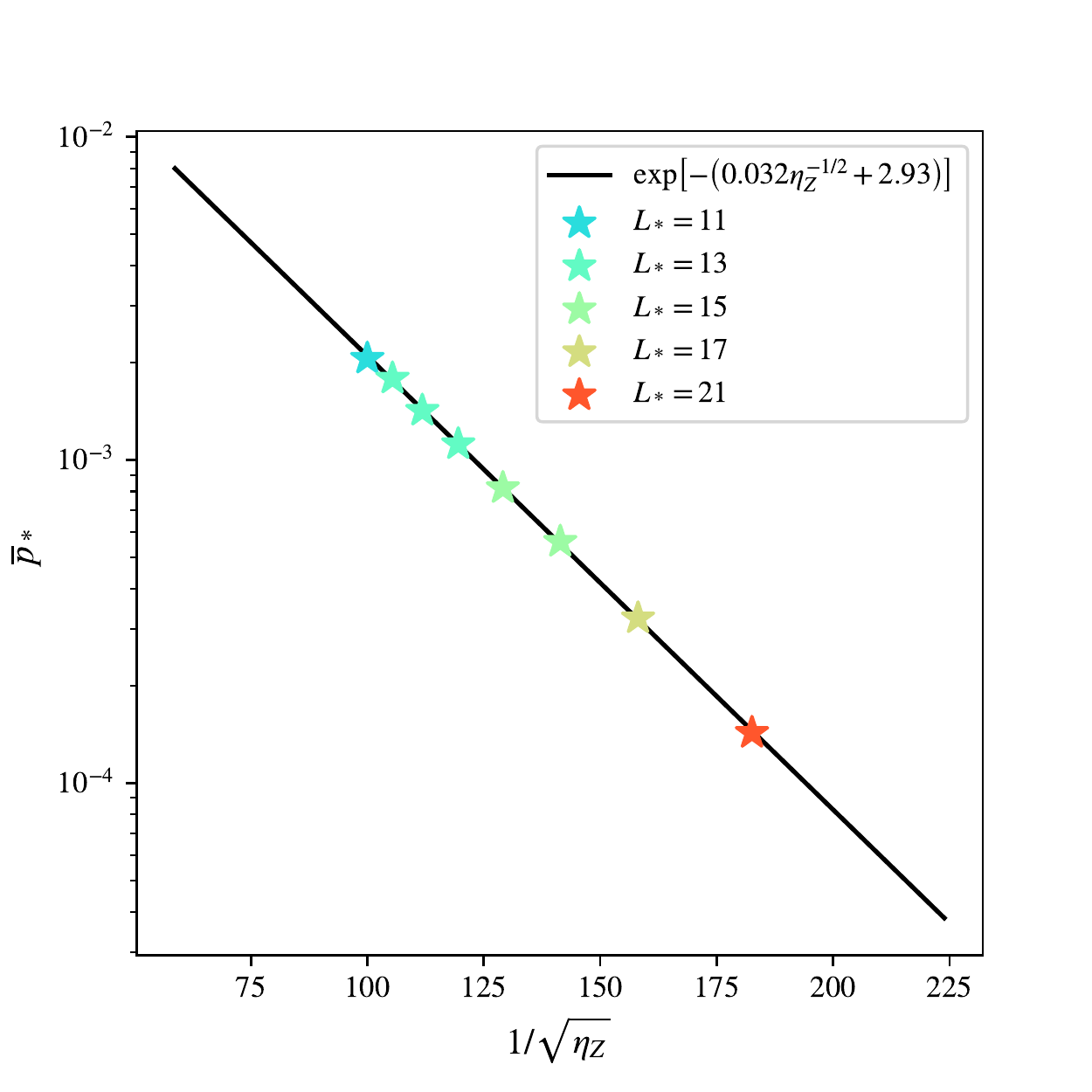}
    }
    
    \subfloat[]{
    \includegraphics[clip, trim=0 1em 0 4em, width=0.95\columnwidth]{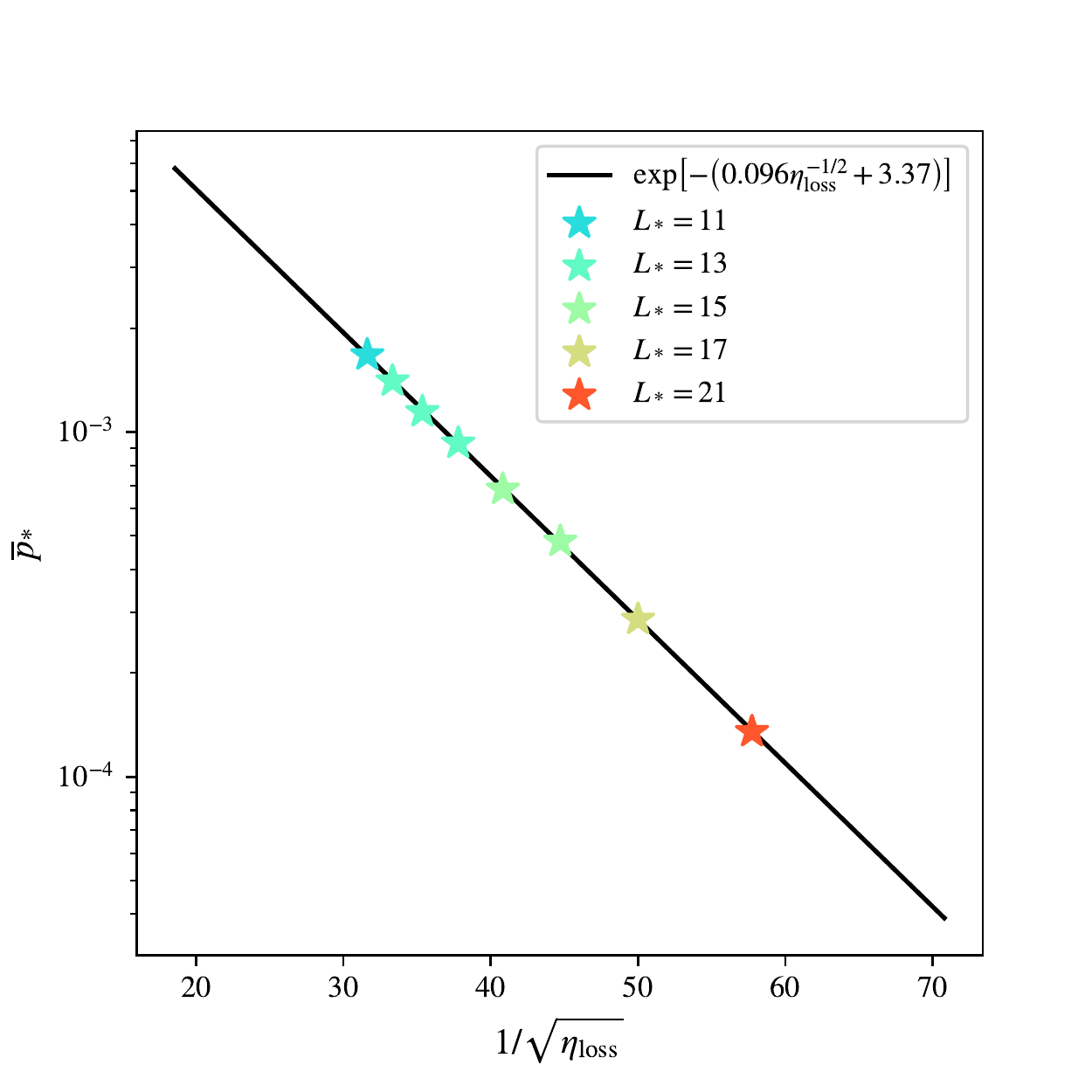}
    }
    
    \caption{(a) $\ln(1/\overline{p}_*)$ vs. $\eta_Z^{-1/2}$ for Error Model~3a (applied to Protocol~\ref{protB}). (b) $\ln(1/\overline{p}_*)$ vs. $\eta_\textrm{loss}^{-1/2}$ for Error Model~3b. Solid lines are fits to Eq.~\eqref{eq:logical_scaling_eta}. Colours indicate the value $L_{*}$ of $L$ that achieves the minimum logical error rate $\overline{p}_{*}$.\label{fig:optimal_vs_eta}}
\end{figure}

From Eq.~\eqref{eq:extrapolation_constants}, we can determine the ``break-even point" beyond which it is advantageous to use the delay lines. That is, since the depolarising noise rate is assumed to be $p = 10^{-3}$, using the experimental setup of Section~\ref{sec:experiment} would make sense only when the logical error rate $\overline{p}$ is below this value. If the delay line error is dominated by dephasing, the break-even point occurs at $\eta_{Z} = 6.5 \times 10^{-5}$. If the delay line error is dominated by loss, the break-even point occurs at $\eta_{\textrm{loss}} = 7.4\times 10^{-4}$. 

As discussed in the following subsection, current experimental estimates for delay line error rates are not below this break-even point. However, the above results show that small reductions in these error rates can lead to very large reductions in the logical error rate. As an example, consider Error Model~3b, in which delay line errors are dominated by qubit loss. For circuit error rates as high as $10^{-3}$, Eqs.~\eqref{eq:logical_scaling_eta} and~\eqref{eq:extrapolation_constants} give logical error rates $\overline{p}_* = 10^{-5}, 10^{-10}, 10^{-15}$ for $\eta_{\textrm{loss}} \approx 1.4 \times 10^{-4}, 2.4 \times 10^{-5}, 9.5 \times 10^{-6}$, respectively. Assuming that $L_{*} \propto \eta_{\textrm{loss}}^{-1/2}$ as above, the values of $L$ required can be estimated to be $L_* \approx 30, 75, 115$, respectively. Thus, even if the circuit error rate is relatively high, with continued improvements in the error rates and storage capacities of delay lines, extremely low logical error rates can potentially be achieved using our protocol. 

\subsection{Experimental prospects}

\begin{table}[]
\begin{tabular}{|c|c|c|}
\hline
{$\tau/T_2$} & {\begin{tabular}[c]{@{}c@{}}optimal code \\ distance ($L_*$)\end{tabular}} & {\begin{tabular}[c]{@{}c@{}}logical error \\ rate ($\overline{p}_*$)\end{tabular}} \\ \hline
$1\times 10^{-5}$ & $21$ & $1.43\times 10^{-4}$ \\ \hline
$1.\overline{3} \times 10^{-5}$ & $17$ & $3.24\times 10^{-4}$ \\ \hline
$1.\overline{6} \times 10^{-5}$ & $15$ & $5.59\times 10^{-4}$ \\ \hline
$2\times 10^{-5}$ & $15$ & $8.20 \times 10^{-4}$ \\ \hline
$2.\overline{3} \times 10^{-5}$ & $13$ & $1.12\times 10^{-3}$ \\ \hline
$2.\overline{6} \times 10^{-5}$ & $13$ & $1.42 \times 10^{-3}$ \\ \hline
$3 \times 10^{-5}$ & $13$ & $1.79 \times 10^{-3}$ \\ \hline
$3.\overline{3}\times 10^{-5}$ & $11$ & $2.06\times 10^{-3}$ \\ \hline
\end{tabular}
\caption{The optimal code distance $L_*$ and corresponding logical error rate $\overline{p}_*$ for different values of $\tau/T_2$ in quantum acoustic systems under Error Model~3a, where $\tau$ is the pulse-to-pulse separation between phonons.}\label{table:etaZ}
\end{table}

We believe that three different experimental platforms are particularly appealing for our purposes: 1) a system of optical photons in a waveguide coupled to an atom or an artificial atom, 2) an integrated superconducting circuit in which single microwave photons can be deterministically generated via a superconducting qubit, 3) a quantum acoustic system based on fabricated nanostructures coupled to a nonlinear quantum emitter, \textit{e.g.}, a transmon qubit piezoelectrically coupled to a phononic waveguide.

In the optical domain, commercially available optical fibers can provide excellent delay times, in principle allowing for an extremely large number of photons in the delay lines. For instance, Tamura \textit{et al.}\ reported a loss rate of $1.4\times 10^{-4}\, \textrm{dB/m}$~\cite{Tamura2018}. Assuming that a single time step lasts $50\,\textrm{ns}$, we obtain a loss rate of $1.4 \times 10^{-3}\,\textrm{dB}$ per time step, which amounts to $\eta_{\textrm{loss}} \approx 9.6 \times 10^{-4}$. This is close to the break-even point $7.4\times 10^{-4}$ estimated above. 

However, weak coupling strengths between a quantum emitter and relevant photon modes can be a limitation of this approach. In particular, the cooperativity $\mathcal{C}$ is the ratio between the probabilities that $\Q$ emits a photon into a guided mode versus into unwanted modes. In order to obtain logical error suppression, $\mathcal{C}$ needs to be sufficiently large, such that the total loss probability is below the loss thresholds found in Section~\ref{sec:threshold_result}. Achieving a high cooperativity, \textit{e.g.}, $\mathcal{C} \gtrsim 100$, is one of the major experimental challenges in the field and is yet to be accomplished. Reducing photon loss at the interfaces of different optical elements and improving the qubit coherence time (in the case of quantum dots) would be another challenge.

In microwave photonics with integrated superconducing circuits, a significantly higher cooperativity $\mathcal{C}\approx 172$ has been achieved~\cite{mirhosseini2019cavity}. 
In fact, more recently, coherent interactions between a quantum emitter and a single time-delayed photon that has propagated through a waveguide have been demonstrated experimentally~\cite{ferreira2020collapse}.
In Ref.~\cite{ferreira2020collapse}, an array of microwave resonators is used to realise a one-dimensional waveguide with delay time $\tau \approx 227\,\textrm{ns}$. The waveguide is coupled to a qubit with photon emission rate $\Gamma_{1D} \approx 2\pi \times 21\,\textrm{MHz}$. This capability implies that around $\tau \Gamma_{1D} \approx 30$ propagating photons can be stored inside the waveguide.
We believe that with further improvements, integrated superconducting circuits can potentially provide a proof-of-principle demonstration of our protocols in the near future. 

Finally, quantum acoustic systems with phononic crystals are also rapidly emerging as a promising platform for our scheme. A single-mode phononic waveguide~\cite{PhysRevLett.121.040501} and an extremely long phonon lifetime ($T_1 \approx 1.5\,\textrm{s}$ and $T_2 \gtrsim 0.3\,\textrm{ms}$)~\cite{maccabe2019phononic} have already been demonstrated in two separate experiments. Assuming that a strong coupling regime with a high cooperativity can be achieved by fabricating integrated nanostructures (similar to superconducting circuits~\cite{mirhosseini2019cavity}), we expect that quantum acoustic systems can realise our protocols for reasonably large system sizes in the near future. For example, with optimistic but reasonable estimates $T_2 \approx 1\,\textrm{ms}$ and $\gamma \approx 100\,\textrm{MHz}$, where $\gamma$ is the coupling strength, one can choose a pulse-to-pulse time separation $\tau \approx 160$~ns to realise high-fidelity gates with error rates below our threshold of $p_\textrm{th}\approx 0.39\%$.
(Here, we assumed that the gate fidelity scales as $1- 1/(\tau \gamma)^2$ based on symmetric wavepackets of phonons~\cite{Pichler2017}.)
This choice\footnote{The logical error rates achievable for other values of $\tau/T_2$ are tabulated in Table~\ref{table:etaZ}.}
of $\tau$ would lead to a delay line error rate per time step of $\eta_{Z}\approx 3\tau/T_2 \approx 4.8 \times 10^{-4}$. Although this is above the break-even point $6.5 \times 10^{-5}$, we note that the experimental technology for quantum acoustic systems is in its early stages and advancing rapidly. Through improving fabrication methods for integrated circuits and lowering the temperature, the coherence times of both qubits and a delay lines may be substantially increased.

\section{Discussion}
\label{sec:discussion}

In this paper, we proposed a method for preparing the well-known three-dimensional cluster state of Ref.~\cite{Raussendorf2006} using a simple experimental setup. The main advantage of our proposal is that it has low {component overhead}, meaning that we only need a handful of experimental components to build a well-protected logical qubit. In contrast, standard protocols based on three-dimensional cluster states or the surface code~\cite{Raussendorf2003,Raussendorf2007,Raussendorf2007a,Fowler2012,Fukui2018} are expected to require hundreds if not thousands of experimental components. 

If memory errors are non-negligible, our protocols do not have finite thresholds for the circuit error rate. Nevertheless, the logical error rate can be made exponentially small in $\eta^{-1/2}$, where $\eta$ represents the memory error rate. Although our estimates suggest that the error rates that have been attained experimentally are not yet small enough, the low component overhead of our approach means that improvements in only a few physical components can lead to extremely large reductions in the logical error rate. 

While the most mature approaches to quantum computation have high component overhead, ours is not the first proposal aiming to reduce component overhead. For example, the promise of anyon-based quantum computation in topological materials~\cite{Kitaev2003,Nayak2008,Nakamura2020} is that natural physical interactions would greatly reduce the component overhead. Likewise, the reader may wonder how our scheme fares in comparison to those based on the Gottesman-Kitaev-Preskill (GKP) code~\cite{Gottesman2001}. This quantum error-correcting code for a qubit in an oscillator was recently used to demonstrate error suppression~\cite{Campagne-Ibarcq2020} by coupling cavity modes that form a GKP code to a transmon. For the protocol used, the logical error rate is determined by (i) a number that decays exponentially with $\sigma^{-2}$, where $\sigma$ is the standard deviation of the Gaussian displacement channel~\cite{Vuillot2019} modelling the dominant source of error on the modes, and (ii) the transmon error rate $p$. The dominant source of error in Ref.~\cite{Campagne-Ibarcq2020} limiting the logical error rate is (ii), leading to a logical error rate that is significantly higher than what (i) might na\"ively suggest. The contribution from (ii) can be reduced to $O(p^2)$ by using a recently proposed fault-tolerant method for preparing GKP states~\cite{Shi2019}, in which case we expect the logical error rate to be limited by $O(p^2)$.

In contrast, in an analogous setting, the logical error rate of our protocols decays exponentially with $\eta^{-{1}/{2}}$, \emph{even if the transmon error rate is significantly higher.} Specifically, it suffices for the transmon error rate to be lower than some threshold value, which we have estimated to be $0.39\%$ in the standard depolarising noise model for circuit errors. Therefore, while approaches based on the GKP code may seem advantageous at the moment, with improved gate fidelities our scheme may be able to outperform them in the future.

From a more theoretical perspective, our protocols have a remarkable fault-tolerance property. Even though there is one physical qubit that interacts with every other qubit during the preparation of the cluster state, the procedure is nonetheless fault-tolerant because any single-qubit circuit-level error results in a constant-weight error on the final state. What is interesting about this phenomenon is that the propagated error may actually be highly nonlocal, yet its \emph{effect} on the specific state we wish to prepare is always the same as that of a geometrically local error. Similarly, the effect of any $m$-qubit circuit-level error on the final state is equivalent to that of at most $m$ geometrically local errors. In fact, this applies not only to our protocols (Protocols~\ref{protA} and~\ref{protB}) for preparing the specific cluster state of Ref.~\cite{Raussendorf2006}, but to our general algorithm (Algorithm~\ref{alg1}), which can be used to prepare the cluster state corresponding to \emph{any} graph. (For general cluster states, geometric locality is defined with respect to the underlying graph; see Appendix~\ref{appendix:error}.) 

By leveraging this fact, we were able to construct fault-tolerant quantum circuits whose depth necessarily scales with the total number of qubits. This is certainly unusual. Fault-tolerant quantum computing protocols usually avoid circuits structured like ours because of the danger that they will spread errors too widely. This often restricts the design of these protocols, leading them to rely on a small number of trusted and manifestly fault-tolerant building blocks, such as transversal gates or ``catch-and-correct''~\cite{Chao2018}. Our work shows that there can be a subtle form of fault-tolerance in which physical errors spread but without adverse effects. This observation may prove useful for generalising our methods to other fault-tolerance schemes. Indeed, Algorithm~\ref{alg1} can immediately be used to generate cluster states obtained by foliating arbitrary stabiliser codes~\cite{Bolt2016,Brown_2020}.\footnote{The fault-tolerance of the resulting protocols can be analysed with the help of Table~\ref{table:errors1} in Appendix~\ref{appendix:errors1}.
}

There are several directions for improvement to explore. For one, the decoder we used in our simulations was the most basic MWPM decoder, which did not take into account matching degeneracies nor the anisotropy of the underlying error model. Decoders that exploit additional information may well obtain better logical error rates and thresholds. Also, there has been a recent flurry of work on using so-called flag techniques to make error-correction schemes more efficient~\cite{Yoder2017,Chao2018,Chao2018a,Chamberland2018}. It would be interesting to investigate whether these techniques could be used to improve our protocols as well. More generally, it could be advantageous to trade a slowly growing component overhead for improved error tolerance. For instance, one could adapt our protocols to build cluster states on an $L\times L\times N$ lattice using $O(L)$ emitters instead of the single emitter we studied here. Such a scheme would still have component overhead parametrically smaller than the $O(L^2)$ physical qubits live in the system. An analogous trade-off was found in~\cite{hilaire2020} between the number of emitters and the entanglement generation rate in the context of quantum communication using cluster states~\cite{Varnava2006,Buterakos2017}. Another strategy would be to concatenate our scheme, replacing our bare single- or dual-rail qubits with qubits protected by error correction, using \textit{e.g.}, GKP~\cite{Gottesman2001} or binomial codes~\cite{Michael2016}. Conversely, our scheme could be used as the inner code, choosing a small value of $L$ to make the logical error rate sufficiently small, \textit{e.g.}, $10^{-5}$. We can then concatenate this with a lower-overhead outer code, which may have a low pseudo-threshold. All of these possibilities are left for future work.

Lastly, we note that our focus lied in the problem of storing a single logical qubit using a small number of experimental components. A natural future direction is to determine how best to perform logical computation using our architecture. A straightforward approach would be to create $n$ logical qubits using $n$ emitters, and implement logical gates between them using lattice surgery techniques~\cite{Horsman2011}. We leave the detailed analysis of such a scheme, as well as the exploration of potentially lower-overhead protocols for fault-tolerant computation, to follow-up work.

\section*{Acknowledgments}

We thank Agnetta Cleland, Sophia Economou, Kevin Multani, Marek Pechal, Hannes Pichler, Amir Safavi-Naeini, Alp Sipahigil, and Zhaoyou Wang for useful discussions. KW is supported by the Stanford Graduate Fellowship. SC acknowledges support from the Miller Institute for Basic Research in Science. IK is supported by the Simons Foundation It from Qubit Collaboration and by the Australian Research Council via the Centre of Excellence in
Engineered Quantum Systems (EQUS) project number CE170100009.
NS is supported by the National Science Foundation Graduate Research Fellowship under Grant No. DGE-1656518.
PH acknowledges the support of AFOSR (award FA9550-19-1-0369), CIFAR, and the Simons Foundation. Our implementation of the minimum-weight perfect matching decoder uses Blossom~V~\cite{Kolmogorov2009}.

\appendix

\section{Alternative algorithm} \label{appendix:alternative}

In this appendix, we present an alternative algorithm, Algorithm~\ref{alg2}, for preparing cluster states $\ket{\psi_G}$ [cf.~Eq.~\eqref{cluster_def}] on arbitrary graphs $G$. Algorithm~\ref{alg2} is similar in structure to Algorithm~\ref{alg1}. The main difference is that unlike Algorithm~\ref{alg1}, Algorithm~\ref{alg2} does not require any intermediate measurements of the ancilla $\Q$ for \textit{any} $G$. 

To understand the form of Algorithm~\ref{alg2}, recall the definition of the graphs $G[k]'$ from Eq.~\eqref{G[k]'}. For each $k \in [n]$, $G[k]'$ consists of the subgraph of $G$ induced by the vertices $[k]$ together with an additional vertex, $\Q$, and an additional edge, $(\Q,k)$. The cluster state $\ket{\psi_{G[k]'}}$ is defined by Eq.~\eqref{cluster_def}. In Appendix~\ref{appendix:correctness_measurement_free}, we prove that for any $k \in [n]$, $\ket{\psi_{G[k]'}}$ is equivalently given by
\begin{widetext}
\begin{equation} \label{psiG[k]'}
    \ket{\psi_{G[k]'}} = \left[\prod_{j=1}^k\left(H_\Q X_{\Q,j}Z_{\Q,j-1}\prod_{i:(i,j) \in E[j]} Z_{\Q,i}\right) \right] \ket{+}_\Q \bigotimes_{i'=1}^k\ket{0}_{i'},
\end{equation}
\end{widetext}
where $E[k] \coloneqq \{(i,j) \in E: i,j \in [k]\}$. Therefore, taking $k = n$ in Eq.~\eqref{psiG[k]'} yields a circuit that recursively generates $\ket{\psi_{G[n]'}}$, as described in the main loop of Algorithm~\ref{alg2}. Once we have prepared $\ket{\psi_{G[n]'}}$, the target cluster state $\ket{\psi_G}$ can then be obtained by applying $Z_{\Q,n}$ (or by measuring $\Q$ in the $Z$-basis).

Observe that if for a given $j$, if $(j-1,j)$ is an edge (Line~\ref{alg2:if} of Algorithm~\ref{alg2}), then 
\begin{equation} \label{cancel Zs} Z_{\Q,j-1} \prod_{i:(i,j) \in E[j]} Z_{\Q,i} = \prod_{\substack{i\neq j-1\\ (i,j) \in E[j]}} Z_{\Q,i},\end{equation}
so we do not apply $Z_{\Q,j-1}$ at all (instead of applying it twice in succession).

Algorithm~\ref{alg2} correctly prepares $\ket{\psi_G}$ for any ordering of the qubits, which is implicitly chosen by labelling the vertices in $V$ from $1$ to $n$. The proof of correctness is given in Appendix~\ref{appendix:correctness_measurement_free}.

\begin{algorithm}[H] \caption{prepare the cluster state $\ket{\psi_G}$ given a graph $G = (V, E)$  (with $V = [n]$)} 
\label{alg2}
\begin{algorithmic}[1]
    \State initialise $\Q$ in $\ket{+}$ \label{alg2: initialise Q}
    \For{$j=1$ to $n$} \State initialise qubit $j$ in 
    $\ket{0}$ \label{alg2: initialise j}
    \If{$(j - 1,j) \in E$} \label{alg2:if}
        \State apply $H_\Q X_{\Q,j}\prod\limits_{\substack{i \neq j-1 \\(i,j) \in E[j]}}Z_{\Q,i}$
    \Else
    \State apply $H_\Q X_{\Q,j} Z_{\Q,j-1}\prod\limits_{i:(i,j) \in E[j]} Z_{\Q,i}$
    
    \textit{\hspace{1em}// the $Z_{\Q,i}$ gates may be applied in any order}
    \EndIf
    \EndFor
    \State apply $Z_{\Q,n}$ \label{alg1:ZQn}
\end{algorithmic}
\end{algorithm}

\section{Correctness}
\label{appendix:correctness}

In this appendix, we prove that Algorithms~\ref{alg1} and~\ref{alg2} correctly prepare cluster states $\ket{\psi_G}$ on arbitrary graphs $G = (V,E)$. We begin with Algorithm~\ref{alg2}.

\subsection{Correctness of Algorithm~\ref{alg2}}
\label{appendix:correctness_measurement_free}

The correctness of Algorithm~\ref{alg2} follows immediately from Eq.~\eqref{psiG[k]'}, which we now prove.
For any graph $G$, the cluster state $\ket{\psi_{G[k]'}}$ corresponding to $G[k]'$ is, by definition [cf.~Eqs.~\eqref{cluster_def} and~\eqref{G[k]'}],
\begin{equation} \label{psiGkdef}
\ket{\psi_{G[k]'}} = Z_{\Q,k}\left[\prod_{(i,j) \in E[k]}Z_{i,j}\right]\ket{+}_\Q\bigotimes_{i'=1}^k\ket{+}_{i'}.
\end{equation}

First, we prove by induction that
\begin{equation}\label{add}  \ket{\psi_{G[k]'}} = \left[\prod_{j=1}^k \textsc{Add}_j\right]\ket{+}_\Q \bigotimes_{i=1}^k\ket{+}_i,\end{equation}
where
\begin{equation} \label{add_def} \textsc{Add}_j \coloneqq Z_{Q,j}\textsc{Swap}_{\Q,j}Z_{\Q, j-1}\prod_{i:(i,j)\in E[j]} Z_{\Q, i}.\end{equation}
Here, $\textsc{Swap}_{a,b}$ denotes the \textsc{Swap} gate between qubits $a$ and $b$. For the base case $k =1$, $Z_{\Q,k-1} = I$ and $E[k] = \varnothing$, so
\begin{align*} 
\textsc{Add}_1\ket{+}_\Q \ket{+}_1 &= Z_{\Q,1}\textsc{Swap}_{\Q,1}\ket{+}_\Q \ket{+}_1 \\
&= Z_{\Q,1}\ket{+}_\Q \ket{+}_q \\
&= \ket{\psi_{G[1]'}} 
\end{align*}
since $G[1]'$ contains only the edge $\{Q,1\} $. Assume that Eq.~\eqref{add} holds for $k-1$. Then, 
\begin{widetext}
\begin{align*}
    \left[\prod_{j=1}^k\textsc{Add}_j\right]\ket{+}_\Q \bigotimes_{i=1}^k\ket{+}_i &= \textsc{Add}_k \ket{\psi_{G[k-1]'}} \ket{+}_k \\
    &= \left[Z_{\Q,k}\textsc{Swap}_{\Q,k} Z_{\Q,k-1}\prod_{l:(l,k) \in E[k]}Z_{\Q,l} \right]\left[Z_{\Q,k-1}\left(\prod_{(i,j) \in E[k-1]}Z_{i,j}\right)\ket{+}_\Q \bigotimes_{i'=1}^{k-1}\ket{+}_{i'}\right]\ket{+}_k \\
    &= Z_{\Q,k} \left[\prod_{l:(l,k)\in E[k]}Z_{k,l}\right]\left[\prod_{(i,j) \in E[k-1]}Z_{i,j}\right]\textsc{Swap}_{\Q,k}\ket{+}_\Q \bigotimes_{i'=1}^k\ket{+}_{i'} \\
    &= Z_{\Q,k}\left[\prod_{(i,j) \in E[k]}Z_{i,j}\right]\ket{+}_\Q \bigotimes_{i'=1}^k\ket{+}_{i'} \\
    &= \ket{\psi_{G[k]'}},
\end{align*}
\end{widetext}
where the third equality uses the identity $\textsc{Swap}_{\Q,k}Z_{\Q,l} = Z_{k,l}\textsc{Swap}_{\Q,k}$ for $l \neq k$, and the fourth equality follows from the fact that $E[k] = E[k-1] \cup \{(l,k) \in E[k]\}$ [cf.~Eq.~\eqref{E[k]}]. 

Next, we observe from Eqs.~\eqref{add_def} and \eqref{add} that for every $j \in [n]$, when $\textsc{Swap}_{\Q,j}$ is applied to $\Q$ and $j$, the qubit $j$ is in the fixed initial state $\ket{+}$. This implies that we do not need to implement a $\textsc{Swap}$ gate that works correctly on arbitrary states. We can instead use an operation that has the same effect when one of the qubits is in the state $\ket{+}$.\footnote{Of course, if the $\textsc{Swap}$ gate were experimentally available, this transformation would be unnecessary. We were primarily motivated to replace the gates in Eq.~\eqref{add} by ones that are more amenable to experimental implementation in the setup considered in Section \ref{sec:experiment}.} In particular, we use the identity
\begin{equation} \label{HQ identity}
Z_{\Q,j}\textsc{Swap}_{\Q,j}\ket{\varphi}_\Q\ket{+}_j = H_\Q X_{\Q,j}\ket{\varphi}_{\Q}\ket{0}_j,
\end{equation}
for any arbitrary state $\ket{\varphi}$ of $\Q$ (potentially entangled with other qubits).
This follows from
\begin{widetext}
\[{\small
\Qcircuit @C=1.1em @R = 0.3em @!R {
&\lstick{\ket{\varphi}} &\qswap &\ctrl{2} &\qw &&&&&&\lstick{\ket{\varphi}} &\ctrl{2} &\qswap &\qw &&&&&&\lstick{\ket{\varphi}} &\qw &\ctrl{2} &\qswap &\qw &&&&&&\lstick{\ket{\varphi}} &\ctrl{2} &\qswap &\gate{H} &\qw \\
&&&&&& = &&&&&&&&& =  &&&&&&&&&& = 
\\
&\lstick{\ket{+}} &\qswap \qwx[-2] &\control\qw &\qw &&&&&&\lstick{\ket{+}} &\control \qw &\qswap \qwx[-2] &\qw &&&&&&\lstick{\ket{0}} &\gate{H} &\control \qw &\qswap \qwx[-2] &\qw &&&&&&\lstick{\ket{0}} &\targ &\qswap \qwx[-2] &\qw &\qw 
\gategroup{1}{28}{3}{31}{1.2em}{--}
} }
\]
\end{widetext}
and the fact that the state prepared in the dashed box is invariant under \textsc{Swap}, for any $\ket{\varphi}$. 
Substituting Eq.~\eqref{HQ identity} into Eq.~\eqref{add}, we arrive at Eq.~\eqref{psiG[k]'}, which forms the basis for Algorithm~\ref{alg2}.

\subsection{Correctness of Algorithm~\ref{alg1}}
\label{appendix:correctness_measurement_based}

To prove the correctness of Algorithm~\ref{alg1},
we show by induction that for every $k \in [n]$, after Line \ref{alg1: Bj} in the $k$th iteration of the \textbf{for} loop has been executed, the state of $\Q$ and the first $k$ data qubits is $\ket{\psi_{G[k]'}}$ [cf.~Eq.~\eqref{psiGkdef}]. 

For the base case $k=1$, $E[1] = \varnothing$ so in Line~\ref{alg1: Bj}, $H_\Q X_{\Q,1}$ is applied to $\Q$ and qubit $1$, which are in their initial states $\ket{+}$ and $\ket{0}$, respectively. By Eq.~\eqref{psiG[k]'} with $k = 1$, this yields the state
\[ H_\Q X_{\Q,1}\ket{+}_\Q\ket{0}_1 = \ket{\psi_{G[1]'}}.\]

For the inductive step, there are two cases to consider, $(k-1,k) \in E$ and $(k-1,k) \notin E$. For both cases, it will be useful to observe from Eq.~\eqref{psiG[k]'} that for any $1 < k \leq n$,
\begin{widetext}
\begin{equation} \label{psiGk recursive} 
\ket{\psi_{G[k]'}} = \left[H_\Q X_{\Q,k}Z_{\Q,k-1}\prod_{i:(i,k) \in E[k]}Z_{\Q,i}\right]\ket{\psi_{G[k-1]'}}\ket{0}_k.
\end{equation}
\end{widetext}
Also note from the definition of $E[j]$ in Eq.~\eqref{E[k]} that the controlled-$Z$ gates applied in Line~\ref{alg1: Bj} can be equivalently written as
\[ \prod_{\substack{i < j -1 \\(i,j) \in E}}Z_{\Q,i} = \prod_{\substack{i \neq j-1 \\ (i,j) \in E[j]}}Z_{\Q,i}. \] 

In the first case $(k-1,k) \in E$, Lines \ref{alg1: if1}--\ref{alg1: reset} in the $(k-1)$th iteration are skipped, so by the inductive hypothesis, the state at the start of the $k$th iteration is $\ket{\psi_{G[k-1]'}}$. Then, Lines \ref{alg1: initialise j} and \ref{alg1: Bj} in the $k$th iteration produce the state
\begin{align*}
\left[H_\Q X_{\Q,k}\prod_{\substack{i \neq k-1\\(i,k) \in E[k]}}Z_{\Q,i} \right]\ket{\psi_{G[k-1]'}}\ket{0}_k = \ket{\psi_{G[k]'}}.
\end{align*}
This follows from Eq.~\eqref{psiGk recursive}, noting that Eq.~\eqref{cancel Zs} holds since $(k-1,k) \in E$.

In the second case $(k-1,k) \not\in E$, the \textbf{if} condition of Line \ref{alg1: if1} is satisfied and $\Q$ is measured in the $Z$-basis. By the inductive hypothesis, $\Q$ and the first $k-1$ data qubits are in the state $\ket{\psi_{G[k-1]'}}$ immediately before this measurement. By Eqs.~\eqref{stabilisers} and~\eqref{G[k]'}, the stabilisers of $\ket{\psi_{G[k-1]'}}$ are generated by $\{X_\Q Z_{k-1}\} \cup \{S_i : i \in [k-1]\}$, where 
\[ S_i = \begin{dcases}X_i \prod_{j:(i,j) \in E[k-1]}Z_j &i < k-1 \\ X_{k-1} Z_{\Q} \prod_{j:(j,k-1) \in E[k-1]} Z_j &i = k-1.
\end{dcases}\]
Hence, the stabiliser generators of post-measurement state of the first $k-1$ data qubits are $\{S_i': i \in [k-1]\}$, where $S_i'= S_i$ for $i <k-1$ and 
\[ S_{k-1}' = \pm X_{k-1}\prod_{j:(j,k-1) \in E[k-1]}Z_j. \]
Here, the $+$ sign corresponds to the post-measurement state of $\Q$ being $\ket{0}$, and the $-$ sign to $\ket{1}$. If the outcome is $\ket{1}$, Line \ref{alg1: Zj} applies $Z_{k-1}$, which negates $S_{k-1}'$ and leaves the other stabiliser generators unchanged. Thus, the state of $\Q$ (which is re-initialised in $\ket{+}$ by Line \ref{alg1: reset}) and the first $k-1$ data qubits at the end of the $(k-1)$th iteration is the cluster state
\begin{equation} \label{endstate2} \left[\prod_{(i,j) \in E[k-1]} Z_{i,j} \right]\ket{+}_\Q\bigotimes_{i'=1}^{k-1}\ket{+}_{i'} = Z_{\Q,k-1}\ket{\psi_{G[k-1]'}}. \end{equation}
Applying Lines~\ref{alg1: initialise j} and~\ref{alg1: Bj} in the $k$th iteration then leads to the state
\[ \left[H_\Q X_{\Q,k}\prod_{\substack{i \neq k-1 \\(i,k) \in E[k]}}Z_{\Q,i}\right]Z_{\Q,k-1}\ket{\psi_{G[k-1]'}}\ket{0}_k = \ket{\psi_{G[k]'}}\]
by Eq.~\eqref{psiGk recursive}.

In both cases, therefore, the state of $\Q$ and the first $k$ data qubits is $\ket{\psi_{G[k]'}}$ at the end of Line \ref{alg1: Bj} of the $k$th iteration, as claimed. In particular, in the $n$th iteration, the state is $\ket{\psi_{G[n]'}}$ at the end of Line \ref{alg1: Bj}. $G[n]'$ differs from $G$ by an extra edge $\{(Q,n)\}$, so by measuring $\Q$ in the $Z$-basis and applying $Z_n$ if necessary in Lines \ref{alg1: measure}--\ref{alg1: Zj}, we obtain the desired state $\ket{\psi_G}$.

\section{Error analysis for arbitrary graphs} \label{appendix:error}

In this appendix, we analyse how errors propagate through Algorithms~\ref{alg1} and~\ref{alg2}, both of which can be used to prepare cluster states $\ket{\psi_G}$ defined by arbitrary graphs $G = (V,E)$. We show that for any $G$, any single-qubit error occurring during either algorithm results in an {effective error} [cf.~Eq.~\eqref{effective error}] whose weight scales with the maximum degree of $G$. For Algorithm~\ref{alg1} and for certain instances of Algorithm~\ref{alg2}, the effective error resulting from a single-qubit error is always \textit{geometrically local}, meaning that it is supported on some subset of $\{i\} \cup N(i)$ for some qubit $i \in V$, where
\begin{equation} \label{N(i)} N(i) \coloneqq \{j:(i,j) \in E\}\end{equation}
denotes the nearest neighbours of $i$ in the graph $G$.\footnote{In the context of quantum error correction, the notion of geometric locality often applies only to graphs that can be embedded into finite-dimensional space. The definition we use here extends to arbitrary graphs, including \textit{e.g.}, expander graphs.} 

A crucial feature of both algorithms---and one of the main intuitions behind the proof below---is that at any point in the procedure, the ancilla qubit $\Q$ is entangled with a restricted number of data qubits (depending on the maximum degree of $G$). Moreover, all or almost all of the qubits with which $\Q$ is entangled at a given point are close to each other in $G$. As a result, even though $\Q$ interacts at least once with every data qubit in $V$, most (and in some cases, all) of the errors that could occur on $\Q$ lead to effective errors on the final state that are localised to neighbourhoods of $G$.

\subsection{Error analysis for Algorithm \ref{alg1}} \label{appendix:errors1}

In this subsection, we prove that for any input graph, single-qubit errors occurring during Algorithm~\ref{alg1} induce geometrically local effective errors. 

\begin{claim} \label{claim1}
Consider an arbitrary graph $G = (V,E)$, and choose any ordering of the qubits in Algorithm~\ref{alg1} by labelling the vertices in $V$ from $1$ to $n$. Then, any single-qubit error occurring between the elementary operations in Algorithm~\ref{alg1} results in effective error [cf.~Eq.~\eqref{effective error}] that is supported on some (possibly empty) subset of $\{i\} \cup N(i)$ for some $i \in [n]$.
\end{claim}

\begin{proof}
For each $j \in [n]$, let
\begin{equation} \label{mathcalBj}
    \mathcal{B}_j \coloneqq H_\Q X_{\Q,j}\prod\limits_{\substack{i< j-1 \\(i,j) \in E}}Z_{\Q,i}
\end{equation}
denote the block of gates applied in the $j$th iteration of the \textbf{for} loop in Algorithm~\ref{alg1} (Line~\ref{alg1: Bj}). We consider the effect of all possible $X$ and $Z$ errors that could occur during Algorithm~\ref{alg1}. (The support of arbitrary errors can then be deduced using Eq.~\eqref{multiple errors}.) Spatially, these errors may inflict the ancilla $\Q$ or one of the data qubits $1, \dots, n$, and temporally, they may occur between two blocks $\mathcal{B}_j$ and $\mathcal{B}_{j+1}$,\footnote{Recall that if $(j, j+1) \not\in E$, $\Q$ is also measured in the $Z$-basis and reset to $\ket{+}$ (Lines~\ref{alg1: measure} and~\ref{alg1: reset}) in the $j$th iteration, between $\mathcal{B}_{j}$ and $\mathcal{B}_{j+1}$. In this case, we consider errors (on $\Q$) both before the measurement (\textit{i.e.}, immediately after $\mathcal{B}_j$) and after the re-initialisation (\textit{i.e.}, immediately before $\mathcal{B}_{j+1}$).} before the first block $\mathcal{B}_1$, after the last block $\mathcal{B}_n$, or between two elementary gates in the same block. We do not consider where the errors occur relative to the $Z_j$ corrections that may be applied for certain $j$ (Line~\ref{alg1: Zj}), as a Pauli error $P$ immediately before $Z_j$ is equivalent to $P$ immediately after $Z_j$ up to a sign. 

\medskip

\noindent \textbf{$Z_i$ errors}

Suppose that a $Z$ error occurs on a data qubit $i \in [n]$. Clearly, $Z_i$ commutes with every operation in Algorithm~\ref{alg1} except for the $X_{\Q,i}$ gate in $\mathcal{B}_i$. This $X_{\Q,i}$ gate is the first gate that acts on qubit $i$, which is initially in the state $\ket{0}$. Therefore, if the $Z_i$ error occurs somewhere before the $X_{\Q,i}$, it has no effect, whereas if it occurs after the $X_{\Q,i}$, it is equivalent to a $Z_i$ error on the final state.

\medskip

\noindent \textbf{$X_\Q$ errors}

We first show that an $X_\Q$ error occurring immediately after $\mathcal{B}_k$ results in an effective error $Z_k$ on the final state, in the sense of Eq.~\eqref{effective error}. This is a consequence of the fact (proven in Appendix~\ref{appendix:correctness_measurement_based}) that immediately after $\mathcal{B}_k$ has been applied, $\Q$ and the first $k$ data qubits are in the cluster state $\ket{\psi_{G[k]'}}$ (and the rest of the qubits are still in their initial state, $\ket{0}^{\otimes n - k})$. Recall from Eq.~\eqref{G[k]'} that $G[k]'$ is the graph with vertices $[k] \cup \Q$ and edges $E[k] \cup \{(\Q,k)\}$. Importantly, $k$ is the only qubit that shares an edge with $\Q$ in $G[k]'$, so it follows from Eq.~\eqref{stabilisers} that $X_\Q Z_k$ is a stabiliser of $\ket{\psi_{G[k]'}}$, which implies
\begin{equation} \label{XtoZ} X_\Q \ket{\psi_{G[k]'}} = Z_k \ket{\psi_{G[k]'}}. \end{equation}
Hence, an $X_\Q$ error immediately after $\mathcal{B}_k$ is equivalent to a $Z_k$ error immediately after $\mathcal{B}_k$, and we know from above that the latter results in a $Z_k$ error on the final state.

Trivially, an $X_\Q$ error occurring at the beginning of the circuit, \textit{i.e.}, before $\mathcal{B}_1$, has no effect since the initial state $\ket{+}$ of $\Q$ is stabilised by $X$. The same goes for $X_\Q$ errors that occur immediately after the re-initialisation of $\Q$ (to $\ket{+}$) in the iterations where $\Q$ is measured.

It remains to consider $X_\Q$ errors that occur between two consecutive gates in the same block $\mathcal{B}_k$. Suppose that an $X_\Q$ error occurs somewhere in $\mathcal{B}_k$ before the $X_{\Q,k}$ gate. At this point, the first $k - 1$ iterations of the \textbf{for} loop have been performed, followed by some subset of the controlled-$Z$ gates in $\mathcal{B}_k$. To be precise, let $J$ denote the subset of qubits $j$ for which $Z_{\Q,j}$ is in $\mathcal{B}_k$ and is located before the $X_\Q$ error. From Eq.~\eqref{mathcalBj}, we have\footnote{The qubits that are included in $J$ depend on the location of the $X_\Q$ error as well as on the order in which the controlled-$Z$ gates in $\mathcal{B}_k$ are actually applied (since they mutually commute, they can be applied in any order).}
\begin{equation} \label{eq:J}
    J \subseteq \{j < k - 1: (j,k) \in E\}.
\end{equation}
As shown in Appendix~\ref{appendix:correctness_measurement_based}, the state of the first $k - 1$ qubits at the end of the $(k-1)$th iteration is $\ket{\psi_{G[k-1]'}}$ if $(k-1,k) \in E$, or $Z_{\Q,k-1}\ket{\psi_{G[k-1]'}}$ if $(k-1,k) \not\in E$ [cf.~Eq.~\eqref{endstate2}]. Consider the case $(k - 1, k) \in E$. The state of the first $k - 1$ qubits at the point where the $X_\Q$ error occurs is then $[\prod_{j \in J}Z_{\Q,j}]\ket{\psi_{G[k-1]'}}$. This state is a cluster state in which $\Q$ shares an edge with $k -1$ [cf.~Eq.~\eqref{G[k]'}] and with every $j \in J$, and is therefore stabilised by $X_\Q Z_{k-1}\prod_{j\in J}Z_j$. It follows that the $X_\Q$ error is equivalent to $Z_{k-1}\prod_{j \in J}Z_j$, which commutes with all subsequent operations in the circuit. The effective error on the final state is therefore $Z_{k-1}\prod_{j\in J}Z_j$.
Similarly, in the case $(k-1,k) \not\in E$, the state at the point where the $X_\Q$ error occurs is $[\prod_{j \in J}Z_{\Q,j}]Z_{\Q,k-1}\ket{\psi_{G[k-1]'}}$. This is stabilised by $X_\Q \prod_{j \in J}Z_j$, so by the same argument, the $X_\Q$ error results in an effective error $\prod_{j \in J}Z_j$.

The only other possibility is that an $X_\Q$ error occurs after the $X_{\Q,k}$ and before the $H_\Q$ in $\mathcal{B}_k$. Since $H_a X_a = Z_a H_a$, this is equivalent to a $Z_\Q$ error occurring after the $H_\Q$, \textit{i.e.}, immediately after $\mathcal{B}_k$. As shown directly below, such an error results in a $Z_{k+1}$ error on the final state if $(k,k+1) \in E$, and no error if $(k,k+1) \not\in E$.

\medskip

\noindent \textbf{$Z_\Q$ errors}

The identities $Z_{a,b}Z_a = Z_a Z_{a,b}$ and $H_a Z_a = X_a H_a$ imply that a $Z_\Q$ error occurring immediately before $\mathcal{B}_k$ or within $\mathcal{B}_k$ (\textit{i.e.}, between any two of the gates in $\mathcal{B}_k$) is equivalent to an $X_\Q$ error immediately after $\mathcal{B}_k$, which in turn results in a $Z_k$ error on the final state [cf.~Eq.~\eqref{XtoZ}].

A $Z_\Q$ error could also occur immediately after $\mathcal{B}_k$. If $(k, k+1) \in E$, this is equivalent to a $Z_\Q$ error immediately before $\mathcal{B}_{k+1}$, and therefore results in a $Z_{k+1}$ error on the final state. On the other hand, if $(k,k+1) \not\in E$, $\Q$ is measured in the $Z$-basis before $\mathcal{B}_{k+1}$ is applied. In this case, the $Z_\Q$ error has no effect since it directly precedes a $Z$-measurement.

\medskip

\noindent \textbf{$X_i$ errors}

Finally, consider the effect of $X_i$ errors. The only gates in Algorithm~\ref{alg1} (besides the $Z_j$ corrections) with which $X_i$ does not commute are the $Z_{\Q,i}$. From Eq.~\eqref{mathcalBj}, we see that a $Z_{\Q,i}$ gate appears in block $\mathcal{B}_j$ for every $j > i + 1$ such that $(i, j) \in E$. Hence, for each $i \in [n]$, define the index set
\begin{equation} \label{eq:I_i}
    I_i \coloneqq \{j > i+1:(i,j) \in E\}
\end{equation}
so that $\mathcal{B}_j$ includes a $Z_{\Q,i}$ gate iff $j \in I_i$.

Suppose that an $X_i$ error occurs before $\mathcal{B}_k$ (but after $\mathcal{B}_{k-1}$) for some $k$. It can be checked using the identities $Z_{a,b}X_b = X_b Z_a Z_{a,b}$ and $H_a X_a = Z_a H_a$ that this is equivalent to an $X_i$ error at the end of the circuit, along with an $X_\Q$ error immediately after each block $\mathcal{B}_j$ for all $j \in I_i$ such that $j \geq k$. As proven above, $X_\Q$ immediately after $\mathcal{B}_j$ results in an effective error $Z_j$ on the final state. Therefore, it follows from Eq.~\eqref{multiple errors} that an $X_i$ occurring immediately after $\mathcal{B}_k$ results in an effective error $X_i \prod_{j \in I_i: j \geq k} Z_j$ on the final state, up to a sign.

\medskip

These results are summarised in Table~\ref{table:errors1}. Note that the effect of any circuit-level $Z$ error is at worst a single-qubit $Z$ error, while the effective error resulting from any $X$ error is (equivalent under stabilisers to) a product of $Z$ errors supported within $N(i)$. It can be verified using Table~\ref{table:errors1} that $X$ and $Z$ errors occurring at the same spacetime location in the circuit result in effective errors supported within $\{i\} \cup N(i)$ \emph{for the same $i$}. This implies via Eq.~\eqref{multiple errors} that \emph{any} single-qubit error at that location leads to an effective error supported within $\{i\} \cup N(i)$, as claimed.

\begin{table*}
\begin{center}
  \begin{tabular}{ c | l | l }
    circuit-level error &location in circuit &effective error on final state \\
    \hhline{=|=|=} 
    \multirow{4}{*}{$X_\Q$} &\multicolumn{1}{l|}{before $\mathcal{B}_1$ or after a re-initialisation of $\Q$}  &none \\
    &\multicolumn{1}{l|}{immediately after the gates $\prod\limits_{j \in J}Z_{\Q,j}$ in $\mathcal{B}_k$}  &$Z_{k-1}\prod\limits_{j\in J}Z_j$ if $(k-1,k) \in E$; $\prod\limits_{j\in J}Z_j$ if $(k-1,k)\not\in E$ \\ 
    &\multicolumn{1}{l|}{between $X_{\Q,k}$ and $H_\Q$ in $\mathcal{B}_k$}  &$Z_{k+1}$ if $(k,k+1) \in E$; none if $(k,k+1)\not\in E$ \\
    &\multicolumn{1}{l|}{immediately after $\mathcal{B}_k$ (\textit{i.e.}, after $H_\Q$ in $\mathcal{B}_k$)} &$Z_k$ \\ \hline
    \multirow{2}{*}{$Z_\Q$} 
    &\multicolumn{1}{l|}{before or within $\mathcal{B}_k$} &$Z_k$ \\
    &\multicolumn{1}{l|}{before a $Z$-measurement of $\Q$} &none \\ \hline
    \multirow{1}{*}{$X_i$} 
    &\multicolumn{1}{l|}{before $\mathcal{B}_k$ (but after $\mathcal{B}_{k-1}$, if $k > 1$)} &$X_i\prod\limits_{j \in I_i:j \geq k}Z_j$ \\
    \hline 
    \multirow{2}{*}{$Z_i$} 
    &\multicolumn{1}{l|}{before $X_{\Q,i}$ (in $\mathcal{B}_i$)} &none\\
    &\multicolumn{1}{l|}{after $X_{\Q,i}$ (in $\mathcal{B}_i$)} &$Z_i$\\
  \end{tabular}\caption{A complete list of $X$ and $Z$ errors that could occur during Algorithm \ref{alg1} (applied to an arbitrary graph $G = (V,E)$) and their effect on the final state [cf.~Eq.~\eqref{effective error}], up to a sign. $\mathcal{B}_j$ is defined in Eq.~\eqref{mathcalBj}. Note from Eqs.~\eqref{eq:J} and~\eqref{eq:I_i} that for each $\mathcal{B}_k$, $J$ is always a subset of $N(k)$, and that for all $i \in [n]$, $I_i \subset N(i)$.
  \label{table:errors1}}
\end{center}
\end{table*}

\end{proof}

Tables~\ref{table:errors_cubic} and~\ref{table:errors_bcc}, which are used in our simulations in Sections~\ref{sec:threshold_result} and~\ref{sec:delay_lineA}, are both special cases of
Table~\ref{table:errors1}. Specifically, Table~\ref{table:errors1} reduces to Table~\ref{table:errors_cubic} (which corresponds to the first step of Protocol~\ref{protA}) for $G = \Gc$ and $\mathcal{B}_j = A_j$ [cf.~Eq.~\eqref{Aj}], and to Table~\ref{table:errors_bcc} (which corresponds to Protocol~\ref{protB}) for $G = \Gbcc$ and $\mathcal{B}_j = B_j$ [cf.~Eq.~\eqref{Bj}].

We remark that unlike Algorithm~\ref{alg1}, not all cluster state preparation circuits have the property that any single-qubit circuit-level error results in a geometrically local effective error. For instance, Algorithm~\ref{alg2} can likewise be used to generate $\ket{\psi_G}$ for any graph $G = (V,E)$. However, as discussed in the following subsection, there exist single-qubit errors in Algorithm~\ref{alg2} that lead to nonlocal effective errors\footnote{More precisely (since multiple operators fit the definition of an effective error resulting from a particular circuit-level error [cf.~Eq.~\eqref{effective error}]), none of the effective errors corresponding to these single-qubit circuit-level errors are geometrically local.} unless the $n$ qubits are ordered such that $(i,i+1) \in E$ for every $i \in [n-1]$. It follows that for any graph that does not contain a Hamiltonian path, Algorithm~\ref{alg2} does not yield---for any possible ordering of the qubits---a preparation circuit for which the effective errors are all geometrically local. As another example, certain single-qubit errors occurring in the circuit given by Equation 2 of Ref.~\cite{Pichler2017}, which prepares a two-dimensional cluster state, would lead to nonlocal errors. This results from the inclusion of several redundant controlled-$Z$ gates. (It should be noted, however, that the experimental protocol proposed in Ref.~\cite{Pichler2017} does not actually apply these redundant gates.) More generally, by adding redundant controlled-$Z$ gates, it is in fact possible to construct circuits in which certain single-qubit errors induce effective errors whose weights necessarily scale with the number of qubits. 

\subsection{Error analysis for Algorithm~\ref{alg2}}

A similar result holds for Algorithm~\ref{alg2}. However, in contrast to Algorithm~\ref{alg1}, the effective errors resulting from single-qubit errors occurring during Algorithm~\ref{alg2} are \textit{not} all geometrically local for every instance.

\begin{claim} \label{claim2}
Consider an arbitrary graph $G = (V,E)$, and choose any ordering of the qubits in Algorithm~\ref{alg2} by labelling the vertices in $V$ from $1$ to $n$. Then, any single-qubit error occurring between the elementary gates in Algorithm~\ref{alg2} results in an effective error [cf.~Eq.~\eqref{effective error}] that is supported on some (possibly empty) subset of $\{i\} \cup N(i) \cup \{i\pm 1\}$, for some $i \in [n]$.
\end{claim}

Specifically, Table~\ref{table:errors2} gives the effect of all possible single-qubit $X$ and $Z$ errors that could occur during Algorithm~\ref{alg2}. In this table, $\widetilde{\mathcal{B}}_j$ is used to denote the $j$th block of gates applied in Algorithm~\ref{alg2}:
\begin{equation} \label{mathcalBjtilde}
    \widetilde{B}_j \coloneqq \begin{dcases} 
    H_\Q X_{\Q,j} \prod_{\substack{i \neq j - 1 \\ (i,j) \in E[j]}}Z_{\Q,i} &(i,i+1) \in E \\
    H_\Q X_{\Q,j} Z_{\Q,j-1}
\prod_{i:(i,j) \in E[j]} Z_{\Q,i} &(i,i+1) \not\in E.    \end{dcases}
\end{equation}

\begin{table*}
\begin{center}
  \begin{tabular}{ c | l | l }
    circuit-level error &location in circuit &effective error on final state \\
    \hhline{=|=|=} 
    \multirow{4}{*}{$X_\Q$} 
    &{before $\widetilde{\mathcal{B}}_1$} &none \\
    &{immediately after the gates $\prod\limits_{j\in \widetilde{J}}Z_{\Q,j}$ in $\widetilde{\mathcal{B}}_k$}  &$Z_{k-1}\prod\limits_{j\in \widetilde{J}}Z_j$ \\ 
    &{between $X_{\Q,k}$ and $H_{\Q}$ in $\widetilde{\mathcal{B}}_k$} &$Z_{k+1}$ \\
    &{immediately after $\widetilde{\mathcal{B}}_k$ (\textit{i.e.}, after $H_\Q$ in $\widetilde{\mathcal{B}}_k$)}  &$Z_k$ \\
    \hline
    \multirow{1}{*}{$Z_\Q$} 
    &{before or within $\widetilde{\mathcal{B}}_k$} &$Z_k$ \\
    \hline
    \multirow{1}{*}{$X_i$} 
    &{before $\widetilde{\mathcal{B}}_k$ (but after $\widetilde{\mathcal{B}}_{k-1}$)} &$X_i\prod\limits_{j\in \widetilde{I}_i:j \geq k}Z_j$ \\ 
    \hline 
    \multirow{2}{*}{$Z_i$} 
    &{before $X_{\Q,i}$ (in $\widetilde{\mathcal{B}}_i$)} &none\\
    &{after $X_{\Q,i}$ (in $\widetilde{\mathcal{B}}_i$)} &$Z_i$\\
  \end{tabular}\caption{A complete list of $X$ and $Z$ errors that could occur in Algorithm~\ref{alg2} and their effect on the final state [cf.~Eq.~\eqref{effective error}], up to a sign. $\widetilde{\mathcal{B}}_j$ is defined in Eq.~\eqref{mathcalBjtilde}. Note from Eqs.~\eqref{eq:tildeI_i} and~\eqref{eq:tildeJ} that $I_i \subset N(i) \cup \{i+1\}$ for all $i \in [n]$, and that for each $\widetilde{\mathcal{B}}_k$, $\widetilde{J}$ is always a subset of $N(k) \cup \{k-1\}$.
  \label{table:errors2}}
\end{center}
\end{table*}

The proof of Claim~\ref{claim2} is essentially the same as that of Claim~\ref{claim1}, requiring only the following modifications. First, there are no intermediate measurements of $\Q$ in Algorithm~\ref{alg2}, so unlike for Algorithm~\ref{alg1} there is no need to consider errors occurring before a measurement or after a re-initialisation of $\Q$. Second, the index sets $J$ and $I_i$ [cf.~Eqs.~\eqref{eq:J} and~\eqref{eq:I_i}] considered in the proof of Claim~\ref{claim1} are slightly different for Claim~\ref{claim2}. Specifically, $I_i$ should be replaced by its analogue $\widetilde{I}_i$, defined as
\begin{equation} \label{eq:tildeI_i}
\widetilde{I}_i \coloneqq \begin{dcases} \{j > i+1:(i,j) \in E\} &(i,i+1) \in E \\
\{i+ 1\} \cup \{j > i:(i,j) \in E\} &(i,i+1)\not\in E, \end{dcases}
\end{equation}
so that the gate block $\widetilde{\mathcal{B}}_j$ includes a $Z_{\Q,i}$ gate iff $j \in \widetilde{I}_i$, as can be seen from Eq.~\eqref{mathcalBjtilde}. (Note that in contrast to $I_i$, which is always contained within the nearest neighbours $N(i)$ of $i$, $\widetilde{I}_i$ may contain $i+1$ even if $(i,i+1) \not\in E$.) Similarly, if $\widetilde{J}$ is a subset of the controlled-$Z$ gates in $\widetilde{\mathcal{B}}_k$, it follows from Eq.~\eqref{mathcalBjtilde} that 
\begin{equation} \label{eq:tildeJ}
\widetilde{J} \subseteq \begin{dcases}
\{j < k-1:(j,k) \in E\} &(k-1,k) \in E \\
\{k-1\} \cup \{j < k: (j,k) \in E\} &(k-1,k) \not\in E.
\end{dcases}
\end{equation}

It can then be verified using Table~\ref{table:errors2} (in conjunction with Eq.~\eqref{multiple errors}) that any single-qubit error occurring in Algorithm~\ref{alg1} results in an effective error supported within $\{i\} \cup N(i) \cup \{i\pm 1\}$. These effective errors are not geometrically local in general, as $i -1$ and $i + 1$ are not necessarily nearest neighbours of $i$.

It is worth noting, however, that the effective errors would always be geometrically local if $(i,i+1) \in E$ for all $i \in [n-1]$. This is possible iff
\begin{enumerate}[1)]
    \item the underlying graph $G$ of the target cluster state contains a Hamiltonian path, and
    \item we use the ordering of the vertices along the Hamiltonian path as the ordering of qubits in Algorithm~\ref{alg2}.
\end{enumerate}
For instance, the cubic lattice $\Gc$ prepared in Protocol~\ref{protA} contains a Hamiltonian path, and the vertices are ordered in Eq.~\eqref{cubic edges} such that $(i,i+1)$ is an edge for every $i \in [n-1]$. For $G = \Gc$
Algorithms~\ref{alg1} and~\ref{alg2} reduce to the exact same circuit when we use this ordering of the vertices, and Table~\ref{table:errors_cubic} shows that all of the effective errors are indeed geometrically local. 

More generally, even if the graph $G$ does not contain a Hamiltonian path, Claim~\ref{claim2} shows that every effective error has weight at most $D(G) + 3$, where $D(G)$ is the maximum degree of $G$, regardless of the ordering of vertices that we choose. 

\bibliography{bib.bib}

\end{document}